\documentclass[journal]{IEEEtran}
\usepackage{graphicx}      
\usepackage{cite}        
\usepackage{siunitx}
\usepackage{amsmath,cancel}
\usepackage{array}
\usepackage{tikz}
\usepackage{mathdots}
\usepackage{yhmath}
\usepackage{cancel}
\usepackage{color}
\usepackage{siunitx}
\usepackage{array}
\usepackage{multirow}
\usepackage{amssymb}
\usepackage{gensymb}
\usepackage{tabularx}
\usepackage{booktabs}
\usepackage{pgfplots}
\usepackage{subcaption}
\usepackage{amsthm}
\usepackage{siunitx}
\pgfplotsset{compat=newest}
\usetikzlibrary{plotmarks}
\usetikzlibrary{arrows.meta}
\usepgfplotslibrary{patchplots}
\usepackage{grffile}
\usepackage{amsmath}
\usetikzlibrary{fadings}
\usetikzlibrary{patterns}
\newtheorem{thm}{Theorem}
\newtheorem{rem}{Remark}

\ifCLASSINFOpdf
\else
\fi
\begin{document}
%
\title{Band-Passing Nonlinearity in Reset Elements}
%
%
%

\author{Nima~Karbasizadeh,~\IEEEmembership{Member,~IEEE,}
        Ali~Ahmadi~Dastjerdi,~\IEEEmembership{Member,~IEEE,}
        Niranjan Saikumar,~\IEEEmembership{Member,~IEEE,}
        and~S.~Hassan~HosseinNia,~\IEEEmembership{Senior~Member,~IEEE}
\thanks{N. Karbasizadeh, A. Ahmadi Dastjerdi, N. Saikumar, S.H.~HosseinNia are with the Department
of Precision and Microsystem Engineering, Delft University of Technology, Delft,
The Netherlands, e-mail: \{N.KarbasizadehEsfahani; A.AhmadiDastjerdi; N.Saikumar; S.H.HosseinNiaKani\}@tudelft.nl.}
\thanks{}
\thanks{}}

%
%

\markboth{Journal of IEEE/ASME Transactions on Mechatronics}%
{}
%



\maketitle

\begin{abstract}
This paper addresses nonlinearity in reset elements and its effects. Reset elements are known for having less phase lag compared to their linear counterparts; however, they are nonlinear elements and produce higher-order harmonics. This paper investigates the higher-order harmonics for reset elements with one resetting state and proposes an architecture and a method of design which allows for band-passing the nonlinearity and its effects, namely, higher-order harmonics and phase advantage. The nonlinearity of reset elements is not entirely useful for all frequencies, e.g., they are useful for reducing phase lag at cross-over frequency region; however, higher-order harmonics can compromise tracking and disturbance rejection performance at lower frequencies. Using proposed ``phase shaping'' method, one can selectively suppress nonlinearity of a single-state reset element in a desired range of frequencies and allow the nonlinearity to provide its phase benefit in a different desired range of frequencies. This can be especially useful for the reset elements in the framework of ``Constant in gain, Lead in phase'' (CgLp) filter, which is a newly introduced nonlinear filter, bound to circumvent the well-known linear control limitation -- the waterbed effect.  
\end{abstract}


%
\IEEEpeerreviewmaketitle

\section{Introduction}
%
%
%
%
\IEEEPARstart{T}{he} growing demand on precision, bandwidth and robustness of controllers in fields like precision motion control are pushing linear control to its limits. Fundamental limits of linear controllers, namely, Bode's phase gain relationship or Bode's sensitivity integral theorem, a.k.a., ``the waterbed effect''~\cite{maclejowski1989multivariate}, has made researchers and industries to change course toward nonlinear control to circumvent these limitations. Reset control is one such nonlinear technique which has gained significant prominence in recent times.\\
Reset control technique was first introduced by Clegg~\cite{clegg1958nonlinear} as a nonlinear integrator and its advantage was described in~\cite{krishnan1974synthesis} in terms of reducing phase lag compared to its linear counterparts. The main idea of reset control is to reset a subset of controller states when a predefined resetting condition is met. More sophisticated reset elements were developed over the years, namely, First-Order Reset Element (FORE)~\cite{horowitz1975non}, Generalized First-Order Reset Element (GFORE)~\cite{guo2009frequency} and Second-Order Reset Element (SORE)~\cite{hazeleger2016second}. These reset elements were used in different capacities such as phase lag reduction, decreasing sensitivity peak, narrowband and broadband phase compensation and approximating the complex-order behaviour~\cite{wu2006reset, li2005nonlinear,li2010reset,li2011optimal,palanikumar2018no,saikumar2019resetting,valerio2019reset}. A new reset-based architecture was recently proposed by~\cite{saikumar2019constant}, which has a constant gain while providing phase lead in a broad range of frequencies. This architecture, named ``Constant in gain, Lead in phase'' (CgLp), can completely replace or take up a significant portion of derivative duties in the framework of PID.\\ 
Being a nonlinear controller, reset elements produce higher-order harmonics, which in turn makes reset control two-edged. While it is capable of overcoming linear control limitations, existence of higher-order harmonics can compromise performance of the system~\cite{karbasizadeh2020benefiting}. Recently researchers found Describing Function (DF) method for analysing the reset elements in frequency domain~\cite{guo2009frequency} insufficient, since it neglects the effect of higher-order harmonics. A generalised form of DF method which accounts for higher-order harmonics called Higher-Order Sinusoidal Input Describing Function (HOSIDF)~\cite{nuij2006higher} was adopted for reset elements in~\cite{saikumar2020loopshaping,dastjerdi2020closed}. There are efforts in the literature to reduce the adverse effects of higher-order harmonics in one frequency or by  tuning of reset element parameters or finding the optimal sequence of elements~\cite{karbasizadeh2020benefiting,bahnamiri2020tuning,dastjerdi2020optimal,cai2020optimal}. However, to the best of authors knowledge, there is no systematic approach in the literature for deliberately reducing higher-order harmonics to a desired upper bound in a range of frequencies.\\
The main benefit of reset elements is reduction of phase lag with respect to their linear counterparts. This characterisation is beneficial in cross-over frequency region and has no clear benefit at other frequencies. Furthermore, higher-order harmonics compromise the performance of the system in terms of tracking precision and disturbance rejection which is basically discussed in loop-shaping method at lower frequencies. Thus, providing a method to band-pass the nonlinearity and in turn higher-order harmonics seems logical to help keep the positive edge of reset elements while limiting its negative edge.\\
The main contribution of this paper is the investigation higher-order harmonics in reset elements with one resetting state and proposing an architecture and a method of design called "phase shaping" to allow for band-passing nonlinearity in reset elements. In other words, using the proposed architecture and phase shaping method one can create a reset element, e.g., a Clegg integrator, a FORE or a CgLp, which is nonlinear in a range of frequencies while it acts linear in terms of steady-state response at other frequencies. Meaning that the element will limit its phase benefits to where it is needed and will have no higher-order harmonics at other frequencies. The paper also investigates the performance of the proposed element in framework of CgLp.\\
The remainder of this paper is organised as follows: The second section presents the preliminaries. The following one introduces and discusses the architecture of the proposed reset element and investigates its HOSIDF. The fourth section will propose the design and tuning method called phase-shaping. The following one will introduce an illustrative example and verify the discussions in practice. Finally, the paper concludes with some remarks and recommendations about ongoing works.
\section{Preliminaries}
This section will discuss the preliminaries of this study.
	\subsection{General Reset Controller}
Following is a general form of a reset controller~\cite{Guo:2015}:
\begin{align}
\label{eq:reset}
{{\sum }_{R}}:=\left\{ \begin{aligned}
& {{{\dot{x}}}_{r}}(t)={{A_r}}{{x}_{r}}(t)+{{B_r}}e(t),&\text{if }e(t)\ne 0\\ 
& {{x}_{r}}({{t}^{+}})={{A}_{\rho }}{{x}_{r}}(t),&\text{if }e(t)=0 \\ 
& u(t)={{C_r}}{{x}_{r}}(t)+{{D_r}}e(t) \\ 
\end{aligned} \right.
\end{align}
where $A_r,B_r,C_r,D_r$ are the state space matrices of the base linear system and $A_\rho=\text{diag}(\gamma_1,...,\gamma_n)$ is called reset matrix. This matrix contains the reset coefficients for each state which are denoted by $\gamma_1,...,\gamma_n$. The controller's input and output are represented by $e(t)$ and $ u(t) $, respectively.
\subsection{$H_\beta$ condition}
The quadratic stability of the closed loop reset system when the base linear system is stable can be examined by the following condition~\cite{beker2004fundamental,guo2015analysis}.
\begin{thm}             
	There exists a constant $\beta \in \Re^{n_r\times 1}$ and positive definite matrix $P_\rho \in \Re^{n_r\times n_r}$, such that the restricted Lyapunov equation
	\begin{eqnarray}
	P > 0,\ &A_{cl}^TP + PA_{cl} &< 0\\
	&B_0^TP &= C_0
	\end{eqnarray}
	has a solution for $P$, where $C_0$ and $B_0$ are defined by
	\begin{align}
	C_0=\left[\begin{array}{ccc}
	\beta C_{p} & 0_{n_r \times n_{nr}} & P_\rho
	\end{array}\right] , & &  B_0=\left[\begin{array}{c}
	0_{n_{p} \times n_{r}}\\
	0_{n_{nr} \times n_{r}}\\
	I_{n_r}
	\end{array}\right].
	\end{align}
	and 
	\begin{equation}
	A_{\rho}^TP_\rho A_{\rho} - P_{\rho} \le 0
	\end{equation}
	where $A_{cl}$ is the closed loop A-matrix.
	$n_r$ is the number of states being reset and $n_{nr}$ being the number of non-resetting states and  ${n_{p}}$ is the number states for the plant.
	$A_p,B_p,C_p,D_p$ are the state space matrices of the plant.
\end{thm}
\subsection{Describing Functions}
Because of its nonlinearity, the steady state response of a reset element to a sinusoidal input is not sinusoidal. Thus, its frequency response should be analysed through approximations like Describing Function (DF) method~\cite{guo2009frequency}. However, the DF method only takes the first harmonic of Fourier series decomposition of the output into account and neglects the effects of the higher order harmonics. As shown in~\cite{karbasizadeh2020benefiting}, this simplification can sometimes be significantly inaccurate. To have more accurate  information about the frequency response of nonlinear systems, a method called ``Higher Order Sinusoidal Input Describing Function'' (HOSIDF) has been introduced in~\cite{nuij2006higher}. This method was developed in \cite{kars2018HOSIDF,dastjerdi2020closed} for reset elements defined by Eq.~(\ref{eq:reset}) as follows:
\begin{align}  \nonumber
\label{eq:hosidf}
& G_n(\omega)=\left\{ \begin{aligned}
& C_r{{(j\omega I-A_r)}^{-1}}(I+j{{\Theta }}(\omega ))B_r+D_r,\quad n=1\\ 
& C_r{{(j\omega nI-A_r)}^{-1}}j{{\Theta }}(\omega )B_r,\quad~~\qquad\text{odd }n> 2\\ 
& 0,\qquad\qquad\qquad\qquad\qquad\qquad\qquad~\text{even }n\ge 2\\ 
\end{aligned} \right. \\
&\begin{aligned}
& {{\Theta }}(\omega )=-\frac{2{{\omega }^{2}}}{\pi }\Delta (\omega )[{{\Gamma }}(\omega )-{{\Lambda }^{-1}}(\omega )] \\  
& \Lambda (\omega )={{\omega }^{2}}I+{{A_r}^{2}} \\  
& \Delta (\omega )=I+{{e}^{\frac{\pi }{\omega }A_r}} \\  
& {{\Delta }_{\rho}}(\omega )=I+{{A}_{\rho}}{{e}^{\frac{\pi }{\omega }A_r}} \\  
& {{\Gamma }}(\omega )={\Delta }_{\rho}^{-1}(\omega ){{A}_{\rho}}\Delta (\omega ){{\Lambda }^{-1}}(\omega ) \\
\end{aligned} 
\end{align}
where $G_n(\omega)$ is the $n^{\text{th}}$ harmonic describing function for sinusoidal input with frequency of $\omega$. \\
According to definition of reset element in open-loop, assuming $e(t)=\sin(\omega t)$, the reset instants will be $t_k=\frac{k \pi }{\omega}$. However, if by changing the architecture of reset element, one can manage to change the resetting condition to the following:
\begin{align}
 \label{eq:reset_phi}
 {{\sum }_{R}}:=\left\{ \begin{aligned}
 & {{{\dot{x}}}_{r}}(t)={{A_r}}{{x}_{r}}(t)+{{B_r}}e(t),&\text{if }\sin(\omega t - \varphi)\ne 0\\ 
 & {{x}_{r}}({{t}^{+}})={{A}_{\rho }}{{x}_{r}}(t),&\text{if }\sin(\omega t - \varphi)=0 \\ 
 & u(t)={{C_r}}{{x}_{r}}(t)+{{D_r}}e(t), \\ 
 \end{aligned} \right.
\end{align}
in other words, if one changes the reset instants, $t_k$,  while maintaining the input itself, the HOSIDF will change to~\cite{dastjerdi2020closed}:
\begin{align}  \nonumber
\label{eq:hosidf_phi}
&G_{\varphi n}(\omega)=\left\{ \begin{aligned}
&C_r{{(A_r-j\omega I)}^{-1}}{{\Theta }_{\varphi}}(\omega )&\\
&\qquad+C_r(j\omega I-A_r)^{-1}B_r+D_r,& n=1\\ 
& C_r{{(A_r-j\omega nI)}^{-1}}{{\Theta }_{\varphi}}(\omega ),&\text{odd }n> 2\\ 
& 0&\text{ even }n\ge 2\\ 
\end{aligned} \right. \\
&\begin{aligned}
&{{\Theta }_{\varphi }}(\omega )=\frac{-2j\omega {{e}^{-j\varphi }}}{\pi } \Omega (\omega ) \left( \omega I\cos (\varphi )+A_r\sin (\varphi ) \right){{\Lambda }^{-1}}(\omega )B\\
&\Omega (\omega )=\Delta (\omega )-\Delta (\omega )\Delta _{\rho}^{-1}(\omega ){{A}_{\rho }}\Delta (\omega ).\\
\end{aligned}
\end{align}
It will be discussed later in the paper that Eq.~(\ref{eq:hosidf_phi}) can be used to obtain the HOSIDF of proposed architecture. 
\subsection{CgLp}
According to~\cite{saikumar2019constant}, CgLp is a broadband phase compensation element whose first harmonic gain behaviour is constant while providing a phase lead. Originally, two architectures for CgLp are suggested using FORE or SORE, both consisting in a reset lag element in series with a linear lead filter, namely ${\sum}_R$ and $D$. For FORE CgLp:
\begin{align}
\label{eq:fore}
&{\sum}_R=\cancelto{A_\rho}{\frac{1}{{s}/{{{\omega }_{r\alpha }}+1}\;}},&D(s)=\frac{{s}/{{{\omega }_{r}}}\;+1}{{s}/{{{\omega }_{f}}}\;+1}
\end{align}
For SORE CgLp:
\begin{equation}
\label{eq:sore}
\begin{aligned}
&{\sum}_R=\cancelto{A_\rho}{\frac{1}{({s}/{{{\omega }_{r\alpha }}{{)}^{2}}+(2s{{{\beta }}}/{{{\omega }_{r\alpha }})}\;+1}\;}}\\ \\
&D(s)=\frac{({s}/{{{\omega }_{r}}{{)}^{2}}+(2s{{{\beta }}}/{{{\omega }_{r}})}\;+1}\;}{({s}/{{{\omega }_{f}}{{)}^{2}}+(2s{{}}/{{{\omega }_{f}})}\;+1}\;}    
\end{aligned}
\end{equation}
In~(\ref{eq:fore}) and~(\ref{eq:sore}), $\omega_{r\alpha}=\omega_r/\alpha$, $\alpha$ is a tuning parameter accounting for a shift in corner frequency of the filter due to resetting action, $\beta$ is the damping coefficient and $[\omega_{r},\omega_{f}]$ is the frequency range where the CgLp will provide the required phase lead. The arrow indicates that the states of element are reset according to $A_\rho$; i.e. are multiplied by $A_\rho$ when the reset condition is met.\\
The main idea behind the CgLp is taking the phase advantage of reset lag element over its linear counter part and use it in combination with a corresponding lead element to create broadband phase lead. Ideally, the gain of the reset lag element should be cancelled out by the gain of the corresponding linear lead element, which creates a constant gain behaviour. The concept is depicted in Fig.~\ref{fig:cglp}.\\
\begin{figure}[t!]
	\centering
	\includegraphics[width=\columnwidth]{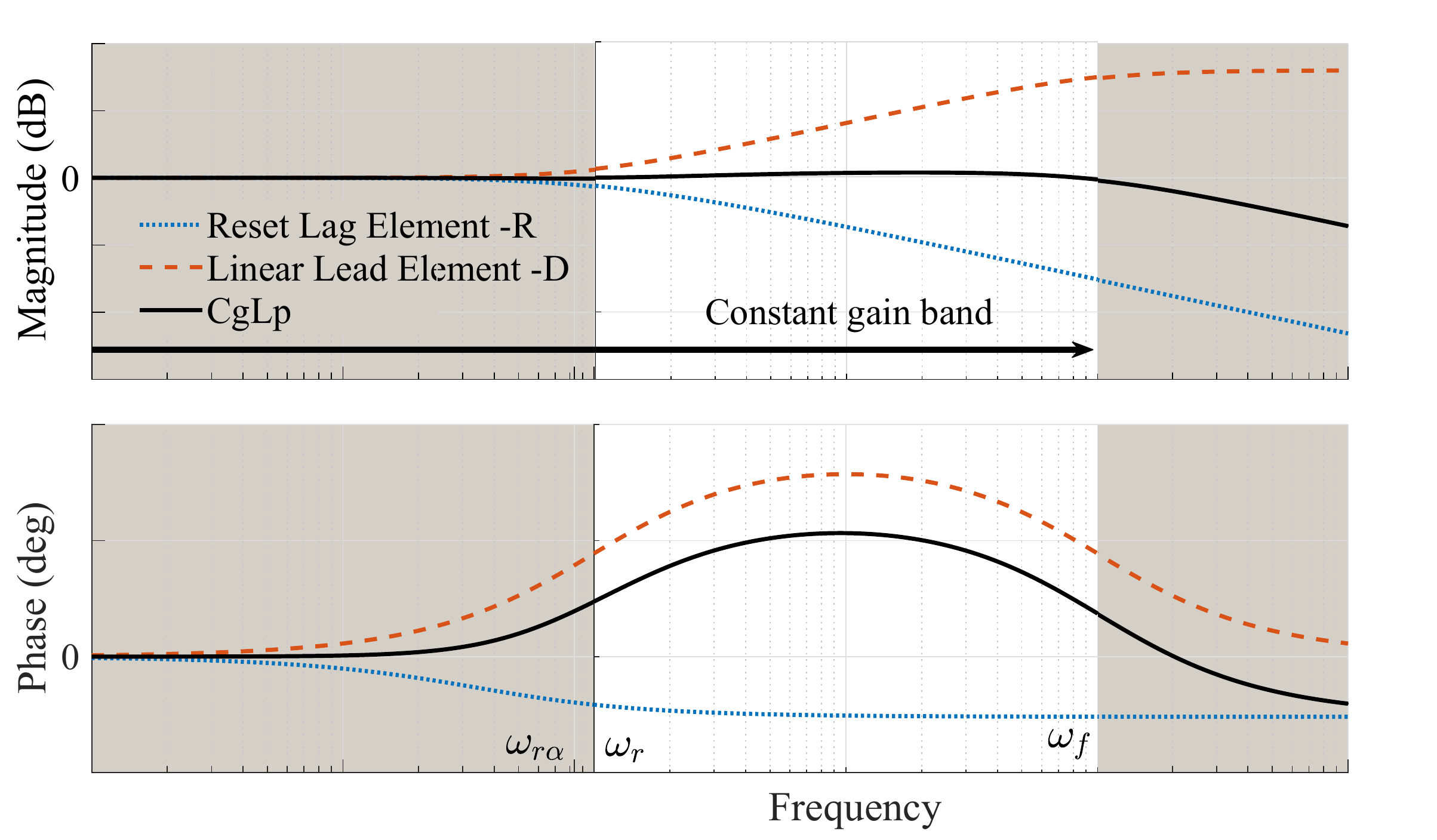}
	\caption{The concept of using combination of a reset lag and a linear lead element to form a CgLp element. The figure is  from~\cite{saikumar2019constant}.}
	\label{fig:cglp}
\end{figure}
\section{A Single State Reset Element Including a Shaping Filter}
This paper proposes an architecture for reset elements with only one resetting state including a shaping filter. This section will analyse the HOSIDF of such an element and its specific properties. The block diagram of the proposed element is presented in Fig.~\ref{fig:ssre_block}.\\
	\begin{figure}[t!]
		\begin{subfigure}{\columnwidth}
			\centering
		\resizebox{0.5\columnwidth}{!}{

			\tikzset{every picture/.style={line width=0.75pt}} 
			
			\begin{tikzpicture}[x=0.75pt,y=0.75pt,yscale=-1,xscale=1]
			
			\draw  [fill={rgb, 255:red, 0; green, 0; blue, 0 }  ,fill opacity=0.05 ][dash pattern={on 4.5pt off 4.5pt}] (57.81,110.95) .. controls (57.81,100.52) and (66.27,92.06) .. (76.7,92.06) -- (180.41,92.06) .. controls (190.84,92.06) and (199.3,100.52) .. (199.3,110.95) -- (199.3,181.11) .. controls (199.3,191.54) and (190.84,200) .. (180.41,200) -- (76.7,200) .. controls (66.27,200) and (57.81,191.54) .. (57.81,181.11) -- cycle ;
			\draw   (109.37,127) -- (163.13,127) -- (163.13,171.22) -- (109.37,171.22) -- cycle ;
			\draw    (98.5,183) -- (171.84,112.09) ;
			\draw [shift={(174,110)}, rotate = 495.96] [fill={rgb, 255:red, 0; green, 0; blue, 0 }  ][line width=0.08]  [draw opacity=0] (10.72,-5.15) -- (0,0) -- (10.72,5.15) -- (7.12,0) -- cycle    ;
			
			\draw    (51.5,147.33) -- (82.49,147.49) -- (106.3,147.94) ;
			\draw [shift={(109.3,148)}, rotate = 181.09] [fill={rgb, 255:red, 0; green, 0; blue, 0 }  ][line width=0.08]  [draw opacity=0] (8.93,-4.29) -- (0,0) -- (8.93,4.29) -- cycle    ;
			\draw  [dash pattern={on 4.5pt off 4.5pt}]  (82.49,147.49) -- (82.49,165.18) ;
			\draw  [dash pattern={on 4.5pt off 4.5pt}]  (82.49,165.18) -- (109.3,165.2) ;
			\draw    (164.3,147) -- (225.5,147.33) ;
			
			\draw (83,134) node    {$e( t)$};
			\draw (187,104) node  [font=\Large]  {$\gamma $};
			\draw (217,133) node    {$u( t)$};
			\draw (136.25,149.5) node [font=\large]  {$\frac{1}{( s/\omega _{r} +1)}$};

			\end{tikzpicture}
		}
			\caption{A conventional FORE}
			\label{fig:fore_block}
		\end{subfigure}
	~\\~\\
	\begin{subfigure}{\columnwidth}
		\centering
		\resizebox{\columnwidth}{!}{

		\tikzset{every picture/.style={line width=0.75pt}} 
		
		\begin{tikzpicture}[x=0.75pt,y=0.75pt,yscale=-1,xscale=1]
		
		\draw  [fill={rgb, 255:red, 0; green, 0; blue, 0 }  ,fill opacity=0.05 ][dash pattern={on 4.5pt off 4.5pt}] (37.81,89.95) .. controls (37.81,79.52) and (46.27,71.06) .. (56.7,71.06) -- (383.41,71.06) .. controls (393.84,71.06) and (402.3,79.52) .. (402.3,89.95) -- (402.3,160.11) .. controls (402.3,170.54) and (393.84,179) .. (383.41,179) -- (56.7,179) .. controls (46.27,179) and (37.81,170.54) .. (37.81,160.11) -- cycle ;
		\draw   (249.37,107) -- (303.13,107) -- (303.13,151.22) -- (249.37,151.22) -- cycle ;
		\draw    (238.5,163) -- (311.84,92.09) ;
		\draw [shift={(314,90)}, rotate = 495.96] [fill={rgb, 255:red, 0; green, 0; blue, 0 }  ][line width=0.08]  [draw opacity=0] (10.72,-5.15) -- (0,0) -- (10.72,5.15) -- (7.12,0) -- cycle    ;
		
		\draw    (31.5,127.33) -- (62.49,127.49) -- (86.3,127.94) ;
		\draw [shift={(89.3,128)}, rotate = 181.09] [fill={rgb, 255:red, 0; green, 0; blue, 0 }  ][line width=0.08]  [draw opacity=0] (8.93,-4.29) -- (0,0) -- (8.93,4.29) -- cycle    ;
		\draw   (89.5,108.33) -- (145.49,108.33) -- (145.49,152.22) -- (89.5,152.22) -- cycle ;
		\draw    (222.3,129) -- (250.3,129) ;
		\draw    (303.3,128) -- (338.3,128) ;
		\draw  [dash pattern={on 4.5pt off 4.5pt}]  (62.49,127.49) -- (62.5,162.86) ;
		\draw  [dash pattern={on 4.5pt off 4.5pt}]  (62.5,162.86) -- (227.5,161.86) -- (227.5,141.86) -- (247.5,141.86) ;
		\draw   (338.5,106.33) -- (385.5,106.33) -- (385.5,150.22) -- (338.5,150.22) -- cycle ;
		\draw    (385.3,128) -- (446.5,128.33) ;
		\draw   (167.37,107) -- (221.13,107) -- (221.13,151.22) -- (167.37,151.22) -- cycle ;
		\draw    (146.3,129) -- (167.3,129) ;
		
		\draw (63,114) node    {$e( t)$};
		\draw (236,117) node    {$x_{2}$};
		\draw (322,86) node [font=\Large]   {$\gamma $};
		\draw (118,128.28) node  [font=\Large]  {$F( s)$};
		\draw (322,116.67) node    {$x_{1}$};
		\draw (427,113) node    {$u( t)$};
		\draw (276.25,130.5) node  [font=\large]  {$\frac{1}{( s/\omega _{r} +1)}$};
		\draw (276,95) node [font=\large] {${\sum}_R$};
		\draw (118,91) node  [font=\footnotesize] [align=left] {\begin{minipage}[lt]{32.674pt}\setlength\topsep{0pt}
			\begin{flushright}
			\textit{Shaping}
			\end{flushright}
			\begin{center}
			\textit{Filter}
			\end{center}
			
			\end{minipage}};
		\draw (194.25,129.11) node  [font=\Large]  {$K( s)$};
		\draw (362,128.28) node  [font=\Large]  {$T( s)$};

		\end{tikzpicture}	
		}
		\caption{A single state reset element including a shaping filter}
		\label{fig:ssre_block}
	\end{subfigure}
	\caption{Block diagrams of a conventional FORE and a single state reset element including a shaping filter proposed by this paper.}
	\label{fig:block_diagrams}
\end{figure}
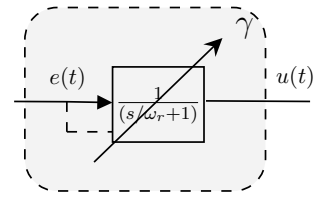
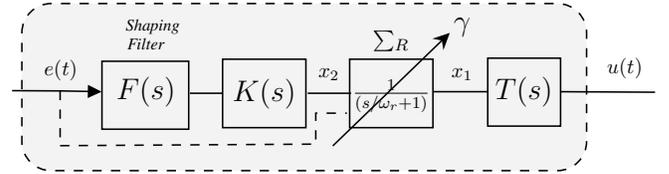
HOSIDF of this element can be found using Eq.~(\ref{eq:hosidf}) with
\begin{equation*}
	{{A}_{\rho }}=\text{diag}(\underbrace{1,...,1}_{{{n}_{F}}+{{n}_{K}}},\gamma ,\underbrace{1,...,1}_{{{n}_{T}}})
\end{equation*} 
where ${{n}_{F}},{{n}_{K}}$ and ${{n}_{T}}$ are number of states for $F(s),K(s)$ and $T(s)$, respectively.\\
However, in this paper, Eq.~(\ref{eq:hosidf_phi}) will be used, since it will reveal more useful information. Let's define:
\begin{equation}
\psi(\omega) :=\angle \frac{X_1(j\omega)}{E(j\omega)}\quad \text{for } \gamma=1.
\end{equation}
It it to be noted that $\psi$ is defined in linear context and is based on base linear system.
\begin{thm}
	\label{thm:psi}
	The higher-order harmonics of the architecture in Fig.~\ref{fig:ssre_block} is an exponential function of $\psi(\omega)$.
	\begin{equation}
		H_n(\omega) = f(n,\omega) (1-e^{-j2\psi})
	\end{equation}
where $f(n,\omega)$ is defined by Eq.~(\ref{eq:g_phi_n}).
\end{thm}
\begin{proof}
	Let's temporarily denote $Q(s) := F(s)K(s)$. For HOSIDF analysis, $e(t)=\sin(\omega t)$; thus,
	\begin{align}
		x_2(t)=\left| Q(j\omega)\right| \sin(\omega t+\phi(\omega))
	\end{align}
	where $\phi(\omega)=\angle Q(j\omega)$. From block diagram of Fig.~\ref{fig:ssre_block} we have
	\begin{equation}
	\label{eq:phi_omega}
		\phi(\omega)=\psi(\omega)+\tan^{-1}(\frac{\omega}{\omega_{r}})
	\end{equation}
	It can be readily seen that input to the ${\sum}_R$ is $x_2(t)$ while the resetting condition is determined by $e(t)$, which have a phase difference. One can find the $n^\text{th}$ harmonic, $H_n(\omega)$, by the following equation:
	\begin{align}
		H_n(\omega)&=Q(j\omega)G_{\varphi n} (\omega)T(jn\omega)\\
		\varphi&=\psi(\omega)+\tan^{-1}(\frac{\omega}{\omega_{r}})
	\end{align}
	where $G_{\varphi n}$ can be obtained using Eq.~(\ref{eq:hosidf_phi}). For $\sum_{R}$, in this paper we have
	\begin{equation}
	\label{eq:r_matrices}
		A_r=-\omega_r, B_r=\omega_{r}, C_r=1,D_r=0, A_\rho=\gamma.
	\end{equation}
	Using Eqs.~(\ref{eq:hosidf_phi}),~(\ref{eq:phi_omega}) and~(\ref{eq:r_matrices}), after some simplifications we have
	\begin{align}
	\label{eq:g_phi_n}
	\begin{aligned}
		{{G}_{\varphi n}}(\omega )&=\left\{ \begin{aligned}
		&f(n,\omega )\left( 1-{{e}^{-j2\psi }} \right) \\
		&\qquad+{{C}_{r}}{{\left( j\omega I-{{A}_{r}} \right)}^{-1}}{{B}_{r}}+{{D}_{r}},  &n=1\\
		&f(n,\omega )\left( 1-{{e}^{-j2\psi }} \right),  &\text{odd }n>2\\
		&0 &\text{even }n\ge 2 \\
		\end{aligned} \right.\\	
		f(n,\omega)&=C_r{{(A_r-j\omega nI)}^{-1}}\frac{\omega {{e}^{-j{{\tan }^{-1}}(\frac{\omega }{{{\omega }_{r}}})}}}{\pi \sqrt{1+\left(\frac{\omega }{{{\omega }_{r}}}\right)^{2}}}\left( 1-\gamma  \right) \delta_\rho^{-1}\delta\\
	\end{aligned}
	\end{align}
	where 
	\begin{align}\nonumber
		\begin{aligned}
			\delta&=\left( 1+{{e}^{\frac{-\pi \omega_r }{{{\omega }}}}} \right)\\
			\delta_\rho&={{\left( 1+\gamma {{e}^{\frac{-\pi \omega_r }{{{\omega }}}}} \right)}}.
		\end{aligned}
	\end{align}
\end{proof}
\begin{rem}
	\label{rem:1}
	Let's define 
	\begin{equation}
		{{\omega }_{lb}}:=\{\omega |\psi (\omega )=0\}.
	\end{equation}
	According to Eqs.~(\ref{eq:g_phi_n}) and~(\ref{eq:hosidf_phi}), 
	\begin{equation}
	G_{\varphi n}(\omega_{lb})=\left\{ \begin{aligned}
	&C(j\omega_{lb} I-A_r)^{-1}B_r+D_r,& n=1\\ 
	& 0&n\ge 2\\ 
	\end{aligned} \right. \\
	\end{equation} 
\end{rem}
Remark~\ref{rem:1} shows that for each frequency in $\omega_{lb}$, all the higher-order harmonics will be zero, in other words, the element will act as its base linear system in terms of steady state output. By way of explanation, when $\psi$ is zero, the reset element, $\sum_{R}$, will reset its state to zero when the state value is already zero. Hence, the resetting action has no effect in steady state response. Figure~\ref{fig:no_phase_shift} depicts this situation. Obviously, in this situation, there exists no phase advantage for the reset element.

\begin{figure}[t!]
	\centering
	\includegraphics[width=\columnwidth]{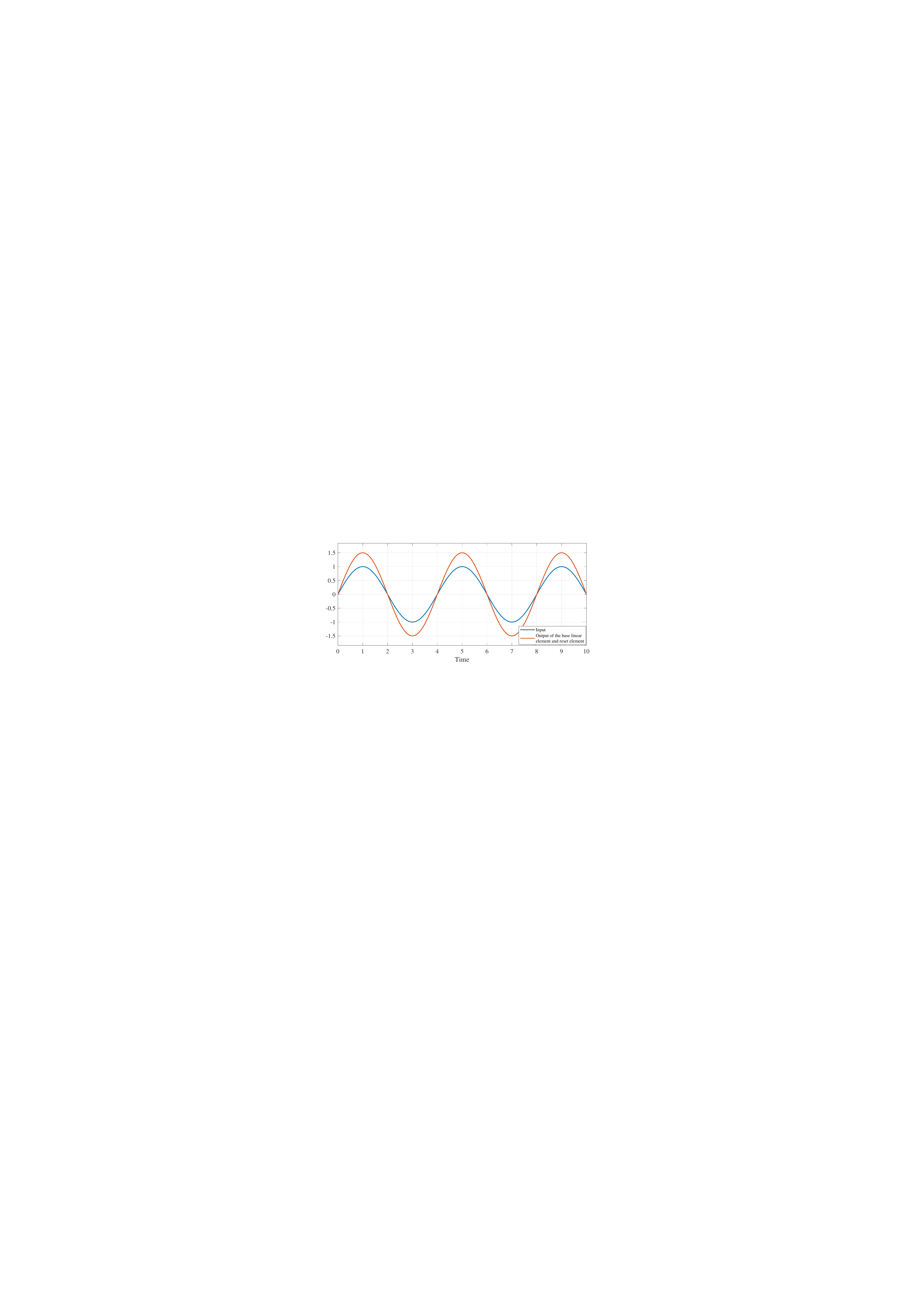}
	\caption{Assuming that the output of the base linear element ($x_1(t)$ for $\gamma=1$) for a reset element has no phase shift with respect to its input ($e(t)$), the output of the reset element itself ($x_1(t)$ for $\gamma\neq 1$) will match the base linear element output at steady state.}
	\label{fig:no_phase_shift}
\end{figure}
\begin{rem}
	For a fixed value of $\omega,\omega_{r}$ and $\gamma$, the maximum of higher-order harmonics magnitude will happen when $\psi(\omega)=\frac{(k+1)\pi}{2}, k\in \mathbb{Z}$.
\end{rem}
\begin{rem}
	\label{rem:phase}
	For $\omega > 10 \omega_r$, the phase of the first-harmonic of the reset element can be approximated by
	\begin{equation}
	\label{eq:phase_psi}
	 \angle {{G}_{\varphi 1}}(\omega )\approx{{\tan }^{-1}}\left( \frac{U\sin(2\psi)-1}{2U\sin^2(\psi)} \right)
	\end{equation}
	where
	\begin{equation}\nonumber
		U=\frac{2  (1-\gamma )}{\pi (1+\gamma )}.
	\end{equation}
	Thus, the phase of the first-harmonic of the reset element for $\omega>10\omega_r$, only depends on $\psi$ and $\gamma$.
\end{rem}
Remark~\ref{rem:phase} implies that at a frequency which is at least one decade higher than $\omega_r$, different combinations of $\psi$ and $\gamma$ may result in the same first-harmonics phase for the reset element. Meanwhile, $\psi$ and $\gamma$ also affect the higher-harmonics magnitude as indicated in Eq.~(\ref{eq:g_phi_n}). Thus solving an optimisation problem, one can find the best combination of $\psi$ and $\gamma$ for a desired first-harmonic phase and minimum higher-order harmonics magnitude.
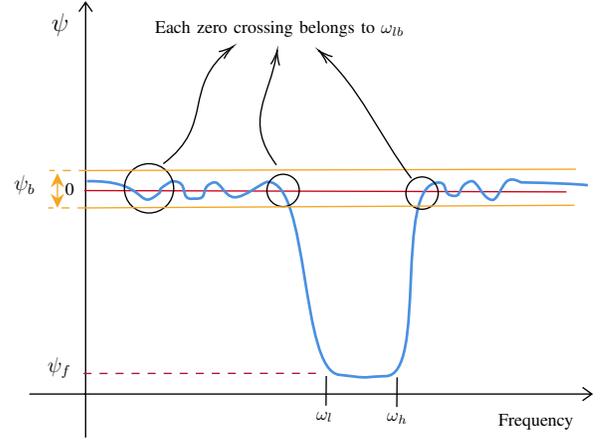
\begin{figure}[t!]
	\centering
	\resizebox{0.9\columnwidth}{!}{
		
		\tikzset{every picture/.style={line width=0.75pt}} 
		
		\begin{tikzpicture}[x=0.75pt,y=0.75pt,yscale=-1,xscale=1]
		
		\draw  (130.79,309.51) -- (551.79,309.51)(172.89,16.11) -- (172.89,342.11) (544.79,304.51) -- (551.79,309.51) -- (544.79,314.51) (167.89,23.11) -- (172.89,16.11) -- (177.89,23.11)  ;
		\draw [color={rgb, 255:red, 208; green, 2; blue, 27 }  ,draw opacity=1 ]   (171.79,157.11) -- (532.79,158.11) ;
		\draw [color={rgb, 255:red, 74; green, 144; blue, 226 }  ,draw opacity=1 ][line width=1.5]    (173.79,150.11) .. controls (217.79,149.11) and (211.86,169.93) .. (223.79,162.11) .. controls (235.72,154.29) and (220.28,160.11) .. (237.28,150.11) .. controls (252.28,149.78) and (241.28,165.11) .. (254.28,163.11) .. controls (266.28,164.11) and (257.26,154.74) .. (268.28,151.11) .. controls (273.79,149.3) and (276.41,161.79) .. (284.28,162.11) .. controls (292.15,162.44) and (305.28,150.6) .. (309.2,150.6) .. controls (340.53,150.6) and (333.3,292.11) .. (362.79,295.11) .. controls (390.28,298.11) and (372.28,297.11) .. (397.28,296.11) .. controls (429.28,296.11) and (402.79,152.11) .. (436.79,151.11) .. controls (445.29,150.86) and (438.52,161.13) .. (449.28,161.11) .. controls (454.66,161.11) and (451.51,152.67) .. (461.28,150.11) .. controls (466.17,148.84) and (469.71,163.59) .. (477.28,164.11) .. controls (481.07,164.38) and (484.09,153.47) .. (488.28,150.11) .. controls (490.37,148.44) and (492.76,148.65) .. (494.79,149.29) .. controls (496.82,149.92) and (498.49,150.98) .. (499.17,150.99) .. controls (520.79,151.36) and (540.29,152.11) .. (548.79,153.11) ;
		\draw   (201.89,155.45) .. controls (201.89,145.26) and (210.15,137) .. (220.34,137) .. controls (230.53,137) and (238.79,145.26) .. (238.79,155.45) .. controls (238.79,165.64) and (230.53,173.91) .. (220.34,173.91) .. controls (210.15,173.91) and (201.89,165.64) .. (201.89,155.45) -- cycle ;
		\draw   (308.59,157.01) .. controls (308.59,150.33) and (314,144.91) .. (320.69,144.91) .. controls (327.37,144.91) and (332.79,150.33) .. (332.79,157.01) .. controls (332.79,163.69) and (327.37,169.11) .. (320.69,169.11) .. controls (314,169.11) and (308.59,163.69) .. (308.59,157.01) -- cycle ;
		\draw   (412.59,159.01) .. controls (412.59,152.33) and (418,146.91) .. (424.69,146.91) .. controls (431.37,146.91) and (436.79,152.33) .. (436.79,159.01) .. controls (436.79,165.69) and (431.37,171.11) .. (424.69,171.11) .. controls (418,171.11) and (412.59,165.69) .. (412.59,159.01) -- cycle ;
		\draw [color={rgb, 255:red, 245; green, 166; blue, 35 }  ,draw opacity=1 ]   (173.79,142.11) -- (539.79,141.11) ;
		\draw [color={rgb, 255:red, 245; green, 166; blue, 35 }  ,draw opacity=1 ]   (172.79,170.11) -- (538.79,168.11) ;
		\draw [color={rgb, 255:red, 245; green, 166; blue, 35 }  ,draw opacity=1 ]   (151.79,147.11) -- (151.79,167.11) ;
		\draw [shift={(151.79,170.11)}, rotate = 270] [fill={rgb, 255:red, 245; green, 166; blue, 35 }  ,fill opacity=1 ][line width=0.08]  [draw opacity=0] (10.72,-5.15) -- (0,0) -- (10.72,5.15) -- (7.12,0) -- cycle    ;
		\draw [shift={(151.79,144.11)}, rotate = 90] [fill={rgb, 255:red, 245; green, 166; blue, 35 }  ,fill opacity=1 ][line width=0.08]  [draw opacity=0] (10.72,-5.15) -- (0,0) -- (10.72,5.15) -- (7.12,0) -- cycle    ;
		\draw [color={rgb, 255:red, 245; green, 166; blue, 35 }  ,draw opacity=1 ] [dash pattern={on 4.5pt off 4.5pt}]  (173.79,142.11) -- (145.79,142.11) ;
		\draw [color={rgb, 255:red, 245; green, 166; blue, 35 }  ,draw opacity=1 ] [dash pattern={on 4.5pt off 4.5pt}]  (172.79,170.11) -- (144.79,170.11) ;
		\draw    (231,137) .. controls (270.6,107.3) and (243.83,79.67) .. (282.1,50.01) ;
		\draw [shift={(283.28,49.11)}, rotate = 503.13] [color={rgb, 255:red, 0; green, 0; blue, 0 }  ][line width=0.75]    (10.93,-3.29) .. controls (6.95,-1.4) and (3.31,-0.3) .. (0,0) .. controls (3.31,0.3) and (6.95,1.4) .. (10.93,3.29)   ;
		\draw    (315.69,137.91) .. controls (290.66,104.62) and (310.66,87.63) .. (316.05,53.68) ;
		\draw [shift={(316.28,52.11)}, rotate = 458.13] [color={rgb, 255:red, 0; green, 0; blue, 0 }  ][line width=0.75]    (10.93,-3.29) .. controls (6.95,-1.4) and (3.31,-0.3) .. (0,0) .. controls (3.31,0.3) and (6.95,1.4) .. (10.93,3.29)   ;
		\draw    (416.69,147.91) .. controls (391.66,114.62) and (378.67,80.16) .. (348.66,53.33) ;
		\draw [shift={(347.28,52.11)}, rotate = 401.05] [color={rgb, 255:red, 0; green, 0; blue, 0 }  ][line width=0.75]    (10.93,-3.29) .. controls (6.95,-1.4) and (3.31,-0.3) .. (0,0) .. controls (3.31,0.3) and (6.95,1.4) .. (10.93,3.29)   ;
		\draw [color={rgb, 255:red, 170; green, 18; blue, 74 }  ,draw opacity=1 ] [dash pattern={on 4.5pt off 4.5pt}]  (171.28,294.11) -- (348,294) ;
		\draw    (353,296) -- (352.8,318.2) ;
		\draw    (406,297) -- (405.8,319.2) ;
		
		\draw (161,156) node   [align=left] {0};
		\draw (510,330) node   [align=left] {Frequency};
		\draw (154,33) node  [font=\Large]  {$\psi $};
		\draw (318,36) node   [align=left] {Each zero crossing belongs to $\displaystyle \omega _{lb}$};
		\draw (143,280.4) node [anchor=north west][inner sep=0.75pt] [font=\large]   {$\psi _{f}$};
		\draw (344,321.4) node [anchor=north west][inner sep=0.75pt]    {$\omega _{l}$};
		\draw (397,322.4) node [anchor=north west][inner sep=0.75pt]    {$\omega _{h}$};
		\draw (118,144.4) node [anchor=north west][inner sep=0.75pt]  [font=\large]  {$\psi _{b}$};

		\end{tikzpicture}
	}
	\caption{Desired shape of phase for $\psi(\omega)$ for band-passing nonlinearity.}
	\label{fig:phase_shaping}
\end{figure}
\section{Phase Shaping Method}
\label{sec:phase_shaping}
Theorem~\ref{thm:psi} and its following remarks, constitute the main idea of phase shaping method for band-passing nonlinearity in reset elements. Previous discussions revealed that the nonlinearity in reset element and its two immediate consequences, namely, phase advantage and higher-order harmonics are dependent on $\psi(\omega)$. The proposed architecture in this paper allows for shaping this phase difference, $\psi$, and consequently nonlinearity in a reset element.\\
From Fig.~\ref{fig:ssre_block}, we have
\begin{equation}
	\psi(\omega)=\angle \left(F(j\omega)K(j\omega)R(j\omega)\right)
\end{equation}
where $R(s)$ represents the base-linear element for ${\sum}_R$. Let
\begin{equation}
	K(s)=\frac{R^{-1}(s)}{s/\omega_{f}+1}=\frac{s/\omega_{r}+1}{s/\omega_{f}+1}
\end{equation}
which is the inverse of the $R(s)$ multiplied to a low-pass filter to make it proper. For a large enough $\omega_{f}$,
\begin{equation}
	\psi(\omega)=\angle F(j\omega).
\end{equation} 
Shaping $\psi(\omega)$ is now reduced to shaping the phase of $F(s)$. If one designs $F(s)$ to have a phase plot as depicted in Fig.~\ref{fig:phase_shaping}, the following will happen: 
\begin{itemize}
	\item Each zero crossing frequency belongs to $\omega_{lb}$, where reset element produces no higher-order harmonics at steady state. These frequencies will be seen as higher-order harmonic notches in HOSIDF.
	\item For frequencies out of $[\omega_l,\omega_h] $, one can upper bound nonlinearities by determining $\psi_{b}$. For a small enough $\psi_{b}$, higher-order harmonics can be approximated to zero. There will be no phase advantage for reset element at these frequencies.
	\item For frequencies in $[\omega_l,\omega_h] $, the reset element will produce high-order harmonics and will have phase advantage.
	\item $\psi_{f}$ and $\gamma$ determine the phase advantage of the reset element in $[\omega_l,\omega_h] $.
\end{itemize}
It can be concluded that by this design, the nonlinearity of the reset element is band-passed in $[\omega_l,\omega_h] $. The details on how to design $F(s)$ to have this phase behaviour will be discussed in the next section. 
\begin{figure}[t!]
	\centering
	\includegraphics[width=\columnwidth]{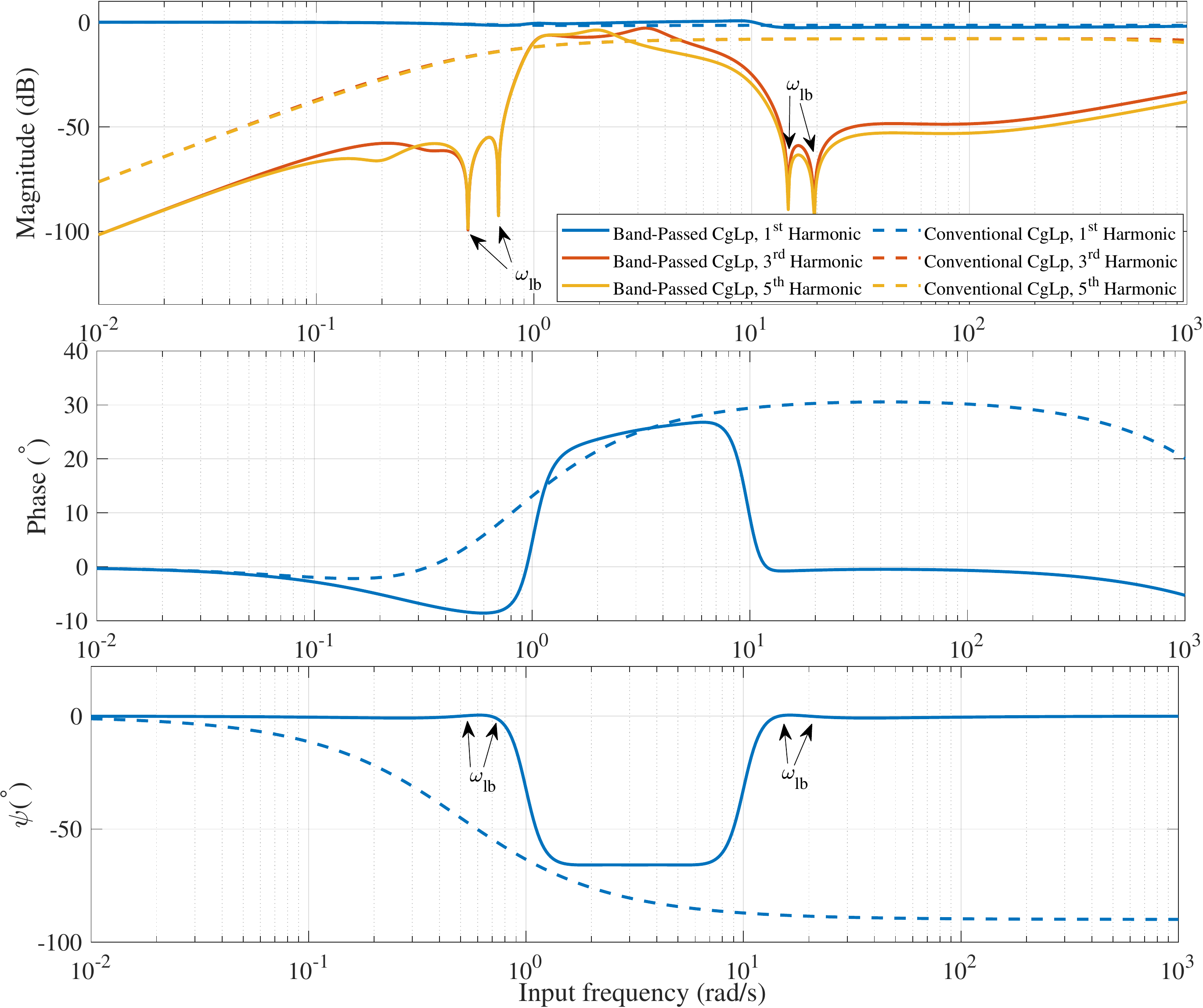}
	\caption{Comparison of HOSIDF of a band-passed CgLp with a conventional one, along with their $\psi$ plot. $\omega_r=0.5$ rad/s and $\gamma$ is 0.2 and 0.35 for band-passed CgLp and conventional one, respectively.}
	\label{fig:bp_cglp}
\end{figure}
\subsection{Band-Passed CgLp}

In order to create a band-passed CgLp, let
\begin{equation}
	T(s)=F^{-1}(s)
\end{equation}
then we have
\begin{equation}
\label{eq:H_n_total}
\begin{aligned}
	&{{H}_{n}}(\omega )=\left\{ \begin{aligned}
	&K(j\omega ){{G}_{\varphi n}}(\omega ),  &n=1\\
	&F(j\omega )K(j\omega ){{G}_{\varphi n}}(\omega ){{F}^{-1}}(jn\omega )  &n>1\\
	\end{aligned} \right.
	\\
	&\varphi(\omega)=\angle F(j\omega) + \tan^{-1}(\frac{\omega}{\omega_{r}})
\end{aligned}
\end{equation}
where $F(s)$ and $K(s)$ should be designed based on the guidelines of the phase shaping method. As aforementioned, it has to be noted that resetting action will cause a shift in corner frequency of $R(s)$~\cite{saikumar2019constant}. In order to account for this frequency shift an additional filter of
\begin{equation}
W(s)=\frac{(s/\omega_{r\alpha}+1)}{(s/\omega_{r}+1)}
\end{equation}  
can be used in $T(s)$.\\
\begin{equation}
	T(s)=F^{-1}(s).W(s)
\end{equation}
In order to verify the discussion, a CgLp element has been band-passed in $[1,10]$ rad/s. The HOSIDF analysis of the band-passed CgLp is compared with a conventional one in Fig.~\ref{fig:bp_cglp}. Both CgLps have $\omega_r=0.5$ rad/s, $\gamma$ is chosen to get approximately same phase advantage. As expected the shaped $\psi$ has made higher-order harmonics zero at frequencies in $\omega_{lb}$ and almost zero at other frequencies out of $[1,10]$ rad/s. The phase advantage is also limited to the band specified.\\ 
It is noteworthy, that by changing the $\psi$ shape in $[\omega_{l},\omega_{h}]$, one can change the shape of phase advantage. This can be useful in creating properties like iso-damping behaviour~\cite{dastjerdi2018tuning,chen2003relay,luo2011experimental}. 
\begin{figure}[t!]
	\centering
	\includegraphics[width=\columnwidth]{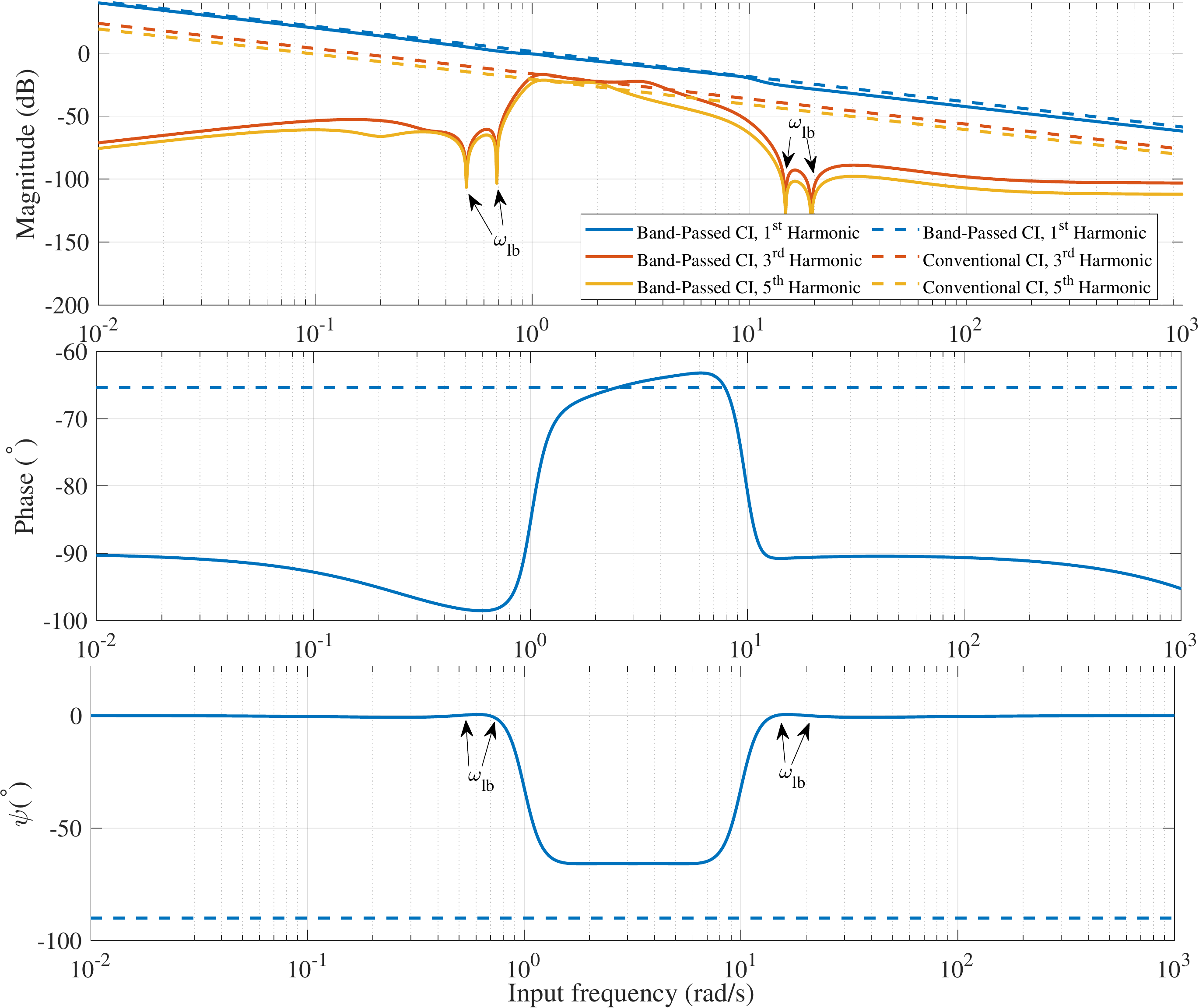}
	\caption{Comparison of HOSIDF of a band-passed Clegg integrator with a conventional one, along with their $\psi$ plot. $\gamma$ is 0.2 and 0.47 for band-passed Clegg Integrator and conventional one, respectively. For band-passed Clegg, $\omega_r=0.5$ rad/s.}
	\label{fig:bp_clegg}
\end{figure}  
\subsection{Band-passed Clegg Integrator and Band-Passed FORE}
Following the same design approach and by letting
\begin{equation}
T(s)=\frac{F^{-1}(s)}{s}
\end{equation}
one can create a band-passed Clegg integrator. Likewise, a band-passed FORE can be created by
\begin{equation}
T(s)=\frac{F^{-1}(s)}{s/\omega_{rr}+1}
\end{equation}
Figures~\ref{fig:bp_clegg} and~\ref{fig:bp_FORE} compares the HOSIDF of a band-passed Clegg integrator and a band-passed FORE with their conventional counterparts. Both reset elements are band-passed in $[1,10]$ rad/s. 
\begin{figure}[t!]
	\centering
	\includegraphics[width=\columnwidth]{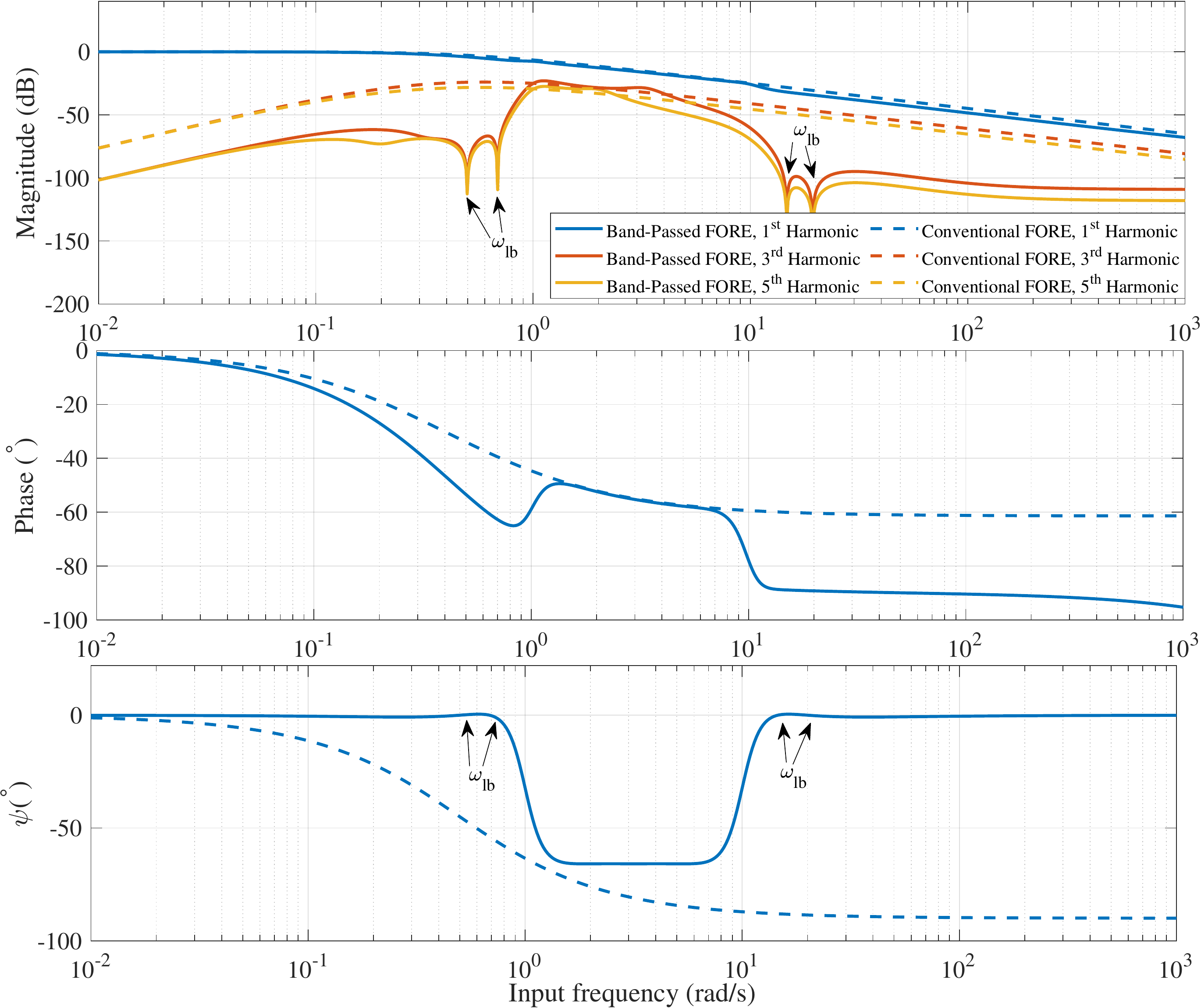}
	\caption{Comparison of HOSIDF of a band-passed CgLp with a conventional one, along with their $\psi$ plot. $\omega_r=0.5$ rad/s and $\gamma$ is 0.2 and 0.4 for band-passed CgLp and conventional one, respectively. For band-passed Clegg, $\omega_{rr}=0.5$ rad/s.}
	\label{fig:bp_FORE}
\end{figure}
As aforementioned, reset elements are usually known and used for their phase lag reduction compared to their linear counterparts~\cite{zaccarian2005first,horowitz1975non,hosseinnia2013fractional}. Phase lag reduction is mainly useful in cross-over frequency region and in other regions it doesn't have a clear benefit. Thus, due to ill-effect of higher-order harmonics especially for tracking and disturbance rejection, the proposed method is useful to band-pass the nonlinearity of these reset elements and its consequent benefits and ill-effects to the cross-over frequency region.\\
\section{Designing the Shaping Filter}
This section proposes a method to design the shaping filter, $F(s)$, such that its phase mimics the schematic shape of Fig.~\ref{fig:phase_shaping}. The first parameter to consider is $\psi_{f}$, which affects the phase of reset element. Assuming the nonlinearity of reset element is being band-passed in cross-over frequency region, if $\omega_r$ is less than $\omega_c/10$, where $\omega_{c}$ is the cross-over frequency, one can choose a combination of $\gamma$ and $\psi_{f}$ according to Eq.~(\ref{eq:phase_psi}) to achieve the desired first-harmonic phase of the reset element at $\omega_c$. \\
The suggested architecture for the shaping filter consists of a lag filter, a notch and an anti-notch filter. The notch and the anti-notch filters should be placed at $\omega_{h}$ and $\omega_{l}$, respectively. Since it is required for phase of the shaping filter to reach a certain value in the passing band, this method requires a flat phase behaviour which is not an integer multiple of $90^\circ$. This is achievable using fractional lag filters; thus, the poles and zeros of the lag filter can be placed according to guidelines of the CRONE approximation of a fractional-order element~\cite{Oustaloup1991}. Such a placement will simplify the calculations.\\
The CRONE approximation is 
\begin{align}
&{{\left( \frac{s/{{\omega }_{l}}+1}{s/{{\omega }_{h}}+1} \right)}^{\lambda }} \approx C \prod_{m=1}^{N} \frac{1+\frac{s}{\omega_{z,m}}}{1+\frac{s}{\omega_{p,m}}}
\label{eq:04:crone} \\
&\omega_{z,m} = \omega_l \left( \frac{\omega_h}{\omega_l} \right)^{\frac{2m-1-\lambda}{2N}}
\label{eq:04:crone:zeros} \\
&\omega_{p,m} = \omega_l \left( \frac{\omega_h}{\omega_l} \right)^{\frac{2m-1+\lambda}{2N}}
\label{eq:04:crone:poles}
\end{align}
where $\lambda \in {{\Re }^{-}}$ and  $N$ is number of poles and zeros. CRONE makes sure that the poles and zeros are placed in equal distance in logarithmic scale. $C$ is the tuning parameter for adjusting the gain of the approximation. However, in this paper, only the phase behaviour of this filter is of interest since the first-order gain behaviour of this element will be cancelled out according to Eq.~(\ref{eq:H_n_total}).\\
The proposed design of shaping filter is
\begin{figure}[t!]
	\centering
	\includegraphics[width=\columnwidth]{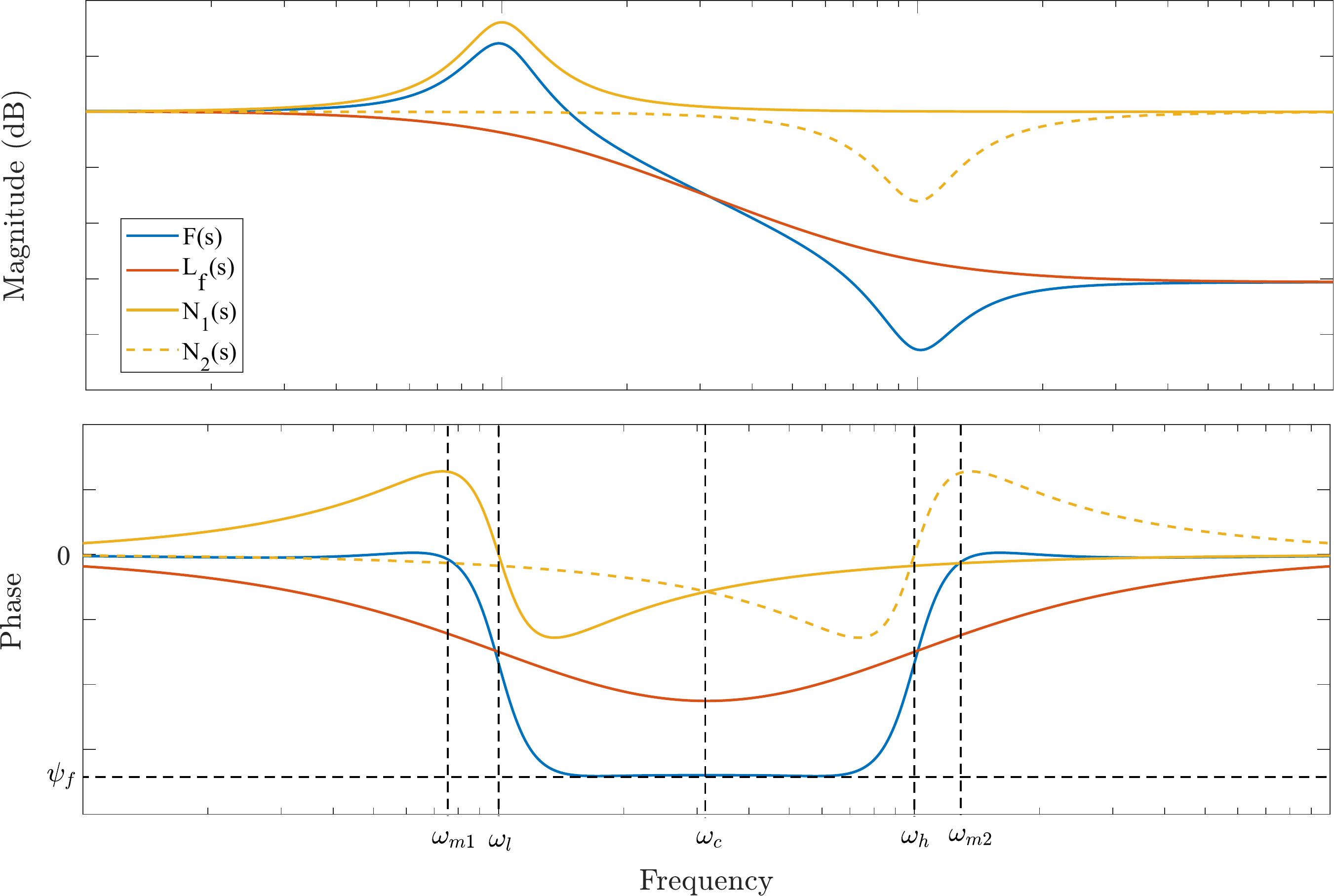}
	\caption{Bode diagram of the composition of $F(s)$.}
	\label{fig:f_s}
\end{figure}  
\begin{align}
& F(s)={{N}_{1}}(s)L_f(s){{N}_{2}}(s) \\ 
&L_f(s)={{\left( \frac{s/{{\omega }_{l}}+1}{s/{{\omega }_{h}}+1} \right)}^{\lambda }}\\
& {{N}_{1}}(s)= \frac{{{\left( s/{{\omega }_{l}} \right)}^{2}}+s/{{\omega }_{l}}+1}{{{\left( s/{{\omega }_{l}} \right)}^{2}}+s/(q{{\omega }_{l}})+1}  \\ 
& {{N}_{2}}(s)= \frac{{{\left( s/{{\omega }_{h}} \right)}^{2}}+s/(q{{\omega }_{h}})+1}{{{\left( s/{{\omega }_{h}} \right)}^{2}}+s/{{\omega }_{h}}+1}  
\end{align}
Thus, there are two parameters to tune, namely, $\lambda$ and $q$. According to Fig.~\ref{fig:f_s} and criteria mentioned in Section~\ref{sec:phase_shaping} for the shaping filter, two constraints can be introduced to find the proper value for  $\lambda$ and $q$. 
\begin{align}
& \angle F(j{{\omega }_{c}})={{\psi }_{f}} \\ 
& \angle F(j{{\omega }_{m1}})=\angle F(j{{\omega }_{m2}})=\varepsilon_1 
\label{eq:constraint_2} 
\end{align}
where $\varepsilon_1$ is a small positive value and
\begin{align}
& \omega_{m1} \in (0,{{\omega }_{l}}) \mid  \frac{d}{d\omega }\angle {{N}_{1}}(j{{\omega }_{m1}})=0\\ 
& \omega_{m2} \in ({{\omega }_{h}},+\infty) \mid \frac{d}{d\omega }\angle {{N}_{2}}(j{{\omega }_{m2}})=0  
\end{align}
Equation (\ref{eq:constraint_2}) ensures that phase of shaping filter remains close to zero and crosses the zero line two time before $\omega_l$ and two times after $\omega_h$.
By symmetry, constraints will be simplified to 
\begin{align}
& \angle {{L}_{f}}(j{{\omega }_{c}})+2\angle {{N}_{1}}(j{{\omega }_{c}})={{\psi }_{f}} \\ 
& \angle {{L}_{f}}(j{{\omega }_{m1}})+\angle {{N}_{1}}(j{{\omega }_{m1}})=\varepsilon_2.
\end{align}
As a rule of thumb, one can choose $\varepsilon_2 =\pi/180$ rad/s.\\
In this paper, without loss of generality, it is assumed that the band-passing range is one decade, i.e., $\omega_{h}=10\omega_{l}$ and the equations are derived in the following.\\ 
Assuming $\omega_{h}=10\omega_{l}$, we have $\omega_{c}=\sqrt{10}\omega_{l}$, thus
\begin{align}
& {{L}_{f}}(j{{\omega }_{c}})\approx \lambda \left( {{\tan }^{-1}}\left( \sqrt{10} \right)-{{\tan }^{-1}}\left( \frac{1}{\sqrt{10}} \right) \right) \\ 
& {{N}_{1}}(j{{\omega }_{c}})={{\tan }^{-1}}\left( \frac{\sqrt{10}}{9q} \right)-{{\tan }^{-1}}\left( \frac{\sqrt{10}}{9} \right) \\ 
& {{L}_{f}}(j{{\omega }_{m1}})=\lambda \left( {{\tan }^{-1}}\left( \zeta  \right)-{{\tan }^{-1}}\left( \zeta /10 \right) \right) \\ 
& {{N}_{1}}(j{{\omega }_{m1}})={{\tan }^{-1}}\left( \frac{\zeta }{1-{{\zeta }^{2}}} \right)-{{\tan }^{-1}}\left( \frac{\zeta /q}{1-{{\zeta }^{2}}} \right)  
\end{align}
where
\begin{equation}
	\zeta =\frac{{{\omega }_{m1}}}{{{\omega }_{l}}}=\frac{\sqrt{2}}{2}\sqrt{\frac{2q+1-\sqrt{1+4q}}{q}}.
\end{equation}
\section{An Illustrative Example}
In order to illustrate the application of the proposed architecture and method in precision motion control, three controllers have been designed and their performance have been compared. The three controllers are a band-passed CgLp, a conventional CgLp designed based on guidelines of~\cite{saikumar2019constant} and a PID. 
\begin{figure}[t!]
	\centering
	\includegraphics[width=\columnwidth]{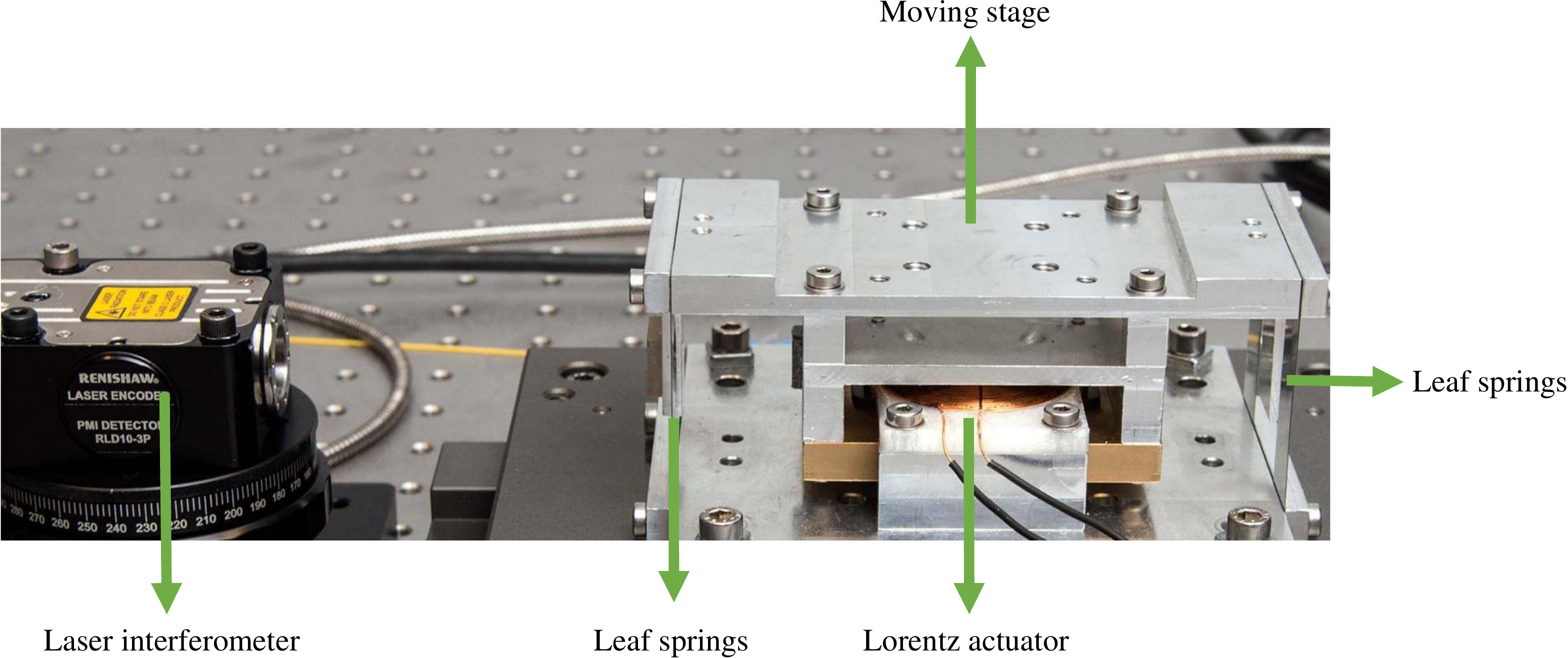}
	\caption{The custom-designed precision stage used for comparison of controllers performance.}
	\label{fig:setup}
\end{figure} 
\subsection{Plant}
The plant which is used for practical implementation is a custom-designed precision stage that is actuated with the use of a Lorentz actuator. This stage is linear-guided using two flexures to attach the Lorentz actuator to the base of the stage and actuated at the centre of the flexures. With a laser encoder, the position of the precision stage is read out with 10 nm resolution. A picture of the setup can be found in Fig.~\ref{fig:setup}. The identified transfer function for the plant is:
\begin{equation}
G(s)=\frac{\num{3.038e4}}{{{s}^{2}}+0.7413s+243.3}
\label{eq:plant}
\end{equation}
Figure~\ref{fig:frf} shows the measured frequency response and that of the identified model.
\begin{figure}[t!]
	\centering
	\includegraphics[width=\columnwidth]{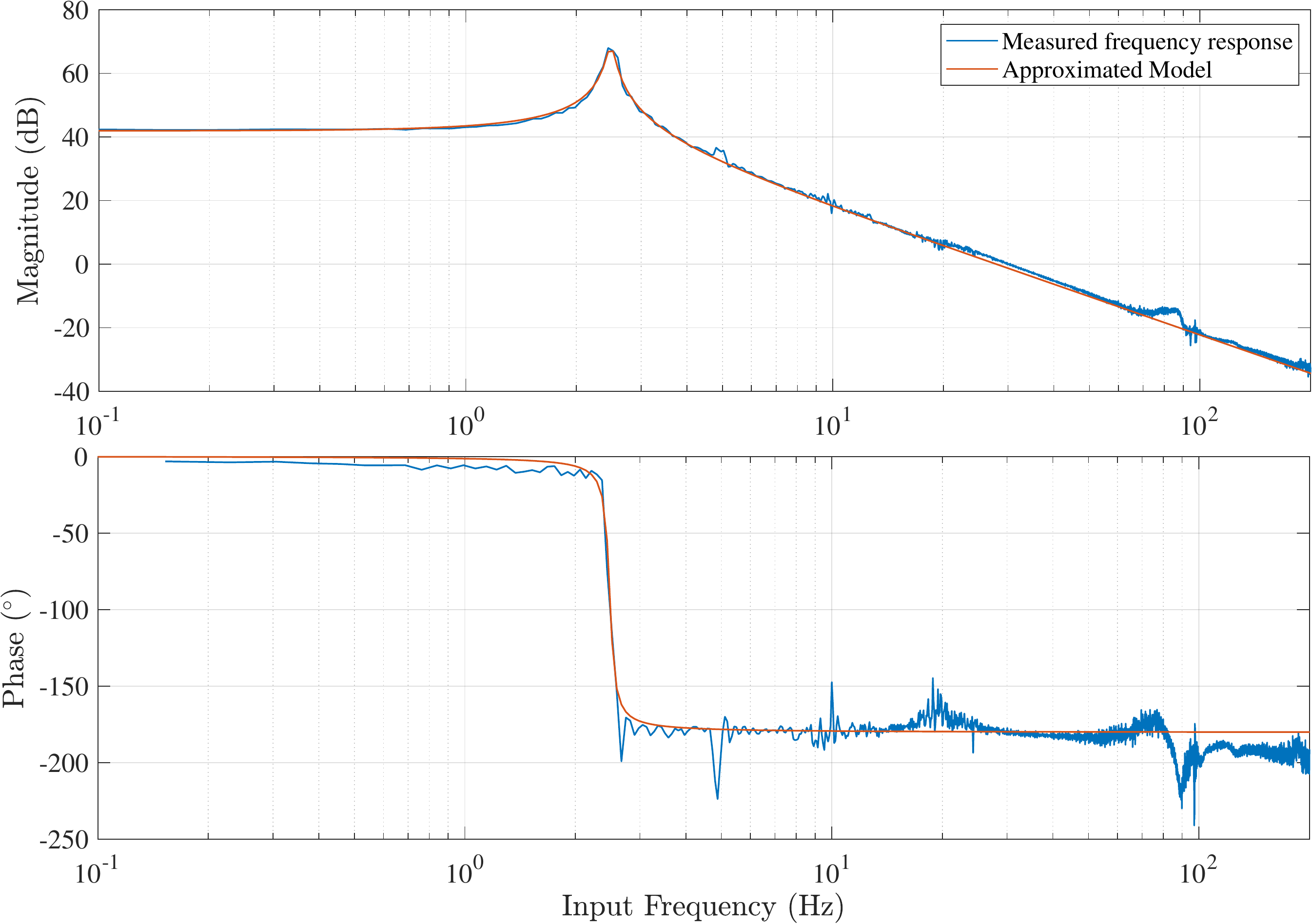}
	\caption{Measured frequency response of the plant and corresponding identified model.}
	\label{fig:frf}
\end{figure} 
\subsection{Controller design approach}
Controllers are designed for a bandwidth of $\omega_c=100$ Hz and phase margin of $42^\circ$. The block diagram of the closed-loop system for CgLps is presented in Fig.~\ref{fig:closed-loop}. The tamed derivative is designed such that the linear part of the controller provides $10^\circ$ of PM for the system and CgLps are designed to provide remaining $32^\circ$. The main reason for existence of the tamed derivative for CgLp controllers is stabilizing the base linear system, which is one of the necessary conditions for stability using $H_\beta$ theorem. For the case of PID controller the whole required PM is provided through tamed derivative.\\
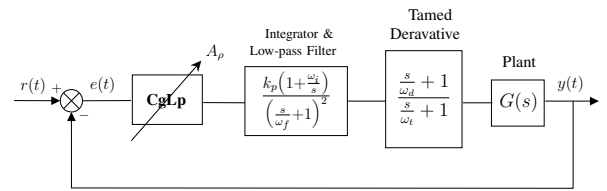
\begin{figure}[t!]
	\centering
	\scalebox{0.9}[0.9]{
		\resizebox{\columnwidth}{!}{
			
			\tikzset{every picture/.style={line width=0.75pt}} 
			
			\begin{tikzpicture}[x=0.75pt,y=0.75pt,yscale=-1,xscale=1]
			
			\draw  [line width=0.75]  (365.21,47.25) -- (440.5,47.25) -- (440.5,143) -- (365.21,143) -- cycle ;
			
			\draw  [line width=0.75]  (115.5,71.17) -- (184.5,71.17) -- (184.5,120) -- (115.5,120) -- cycle ;
			\draw [line width=0.75]    (114.5,136.08) -- (181.53,59.34) ;
			\draw [shift={(183.5,57.08)}, rotate = 491.13] [fill={rgb, 255:red, 0; green, 0; blue, 0 }  ][line width=0.08]  [draw opacity=0] (10.72,-5.15) -- (0,0) -- (10.72,5.15) -- (7.12,0) -- cycle    ;
			
			\draw  [line width=0.75]  (226.5,63) -- (327.5,63) -- (327.5,131) -- (226.5,131) -- cycle ;
			
			\draw [line width=0.75]    (66.42,96.67) -- (114.22,96.45) ;
			\draw [line width=0.75]    (550,97) -- (550,183) -- (55.5,182.67) -- (55.46,110.73) ;
			\draw [shift={(55.46,107.73)}, rotate = 449.97] [fill={rgb, 255:red, 0; green, 0; blue, 0 }  ][line width=0.08]  [draw opacity=0] (8.93,-4.29) -- (0,0) -- (8.93,4.29) -- cycle    ;
			\draw [line width=0.75]    (-1,96.45) -- (41.5,96.65) ;
			\draw [shift={(44.5,96.67)}, rotate = 180.27] [fill={rgb, 255:red, 0; green, 0; blue, 0 }  ][line width=0.08]  [draw opacity=0] (8.93,-4.29) -- (0,0) -- (8.93,4.29) -- cycle    ;
			\draw [line width=0.75]    (522.5,97) -- (572.5,97) ;
			\draw [shift={(575.5,97)}, rotate = 180] [fill={rgb, 255:red, 0; green, 0; blue, 0 }  ][line width=0.08]  [draw opacity=0] (8.93,-4.29) -- (0,0) -- (8.93,4.29) -- cycle    ;
			\draw [line width=0.75]    (185.42,98.45) -- (226.5,98.67) ;
			\draw [line width=0.75]    (327.5,97) -- (365.5,97) ;
			\draw  [line width=0.75]  (469.5,76.17) -- (521.5,76.17) -- (521.5,121.67) -- (469.5,121.67) -- cycle ;
			\draw [line width=0.75]    (441.5,98.67) -- (470.5,98.67) ;
			\draw   (44.5,96.67) .. controls (44.5,90.56) and (49.41,85.61) .. (55.46,85.61) .. controls (61.52,85.61) and (66.42,90.56) .. (66.42,96.67) .. controls (66.42,102.77) and (61.52,107.73) .. (55.46,107.73) .. controls (49.41,107.73) and (44.5,102.77) .. (44.5,96.67) -- cycle ; \draw   (47.71,88.85) -- (63.21,104.49) ; \draw   (63.21,88.85) -- (47.71,104.49) ;
			
			\draw (17.84,81.76) node  [font=\large]  {$r( t)$};
			\draw (547.62,79.76) node  [font=\large]  {$y( t)$};
			\draw (86.35,80.76) node  [font=\large]  {$e( t)$};
			\draw (496.57,97.58) node  [font=\Large]  {$G( s)$};
			\draw (497,58) node  [font=\normalsize] [align=left] {{\fontfamily{ptm}\selectfont {\large Plant}}};
			\draw (41,85) node    {$+$};
			\draw (68,111) node    {$-$};
			\draw (402.85,95.13) node  [font=\Large]  {$\dfrac{\frac{s}{\omega _{d}} +1}{\frac{s}{\omega _{t}} +1}$};
			\draw (404,21) node  [font=\large] [align=left] {{\fontfamily{ptm}\selectfont \textbf{ \ \ }Tamed \ }\\{\fontfamily{ptm}\selectfont Deravative}};
			\draw (197,48.91) node  [font=\large]  {$A_{\rho }$};
			\draw (149,96.58) node   [align=left] [font=\large] {\textbf{CgLp}};
			\draw (277,97) node  [font=\LARGE]  {$\frac{k_{p}\left( 1+\frac{\omega _{i}}{s}\right)}{\left(\frac{s}{\omega_f}+1\right)^2}$};
			\draw (275,40) node  [font=\normalsize] [align=left]{{\fontfamily{ptm}\selectfont \textbf{ \ \ }Integrator \& \ }\\{\fontfamily{ptm}\selectfont Low-pass Filter}};

			\end{tikzpicture}
			
		}
	}
	\caption{Designed control architecture to compare the performance of two sets of controllers.} 
	\label{fig:closed-loop}
\end{figure} 
Table~\ref{tab:params} shows the parameters for the designed controllers. Figure~\ref{fig:ol_hosidf} shows the open loop HOSIDF analysis for them including the plant. As expected CgLp controllers show a higher first-order harmonic gain than PID in lower frequencies while they have the same phase margin as PID. However, due to the design method presented in this paper, the band-passed CgLp shows significant decrease in higher-order harmonics than conventional one. Consequently, one can expect an improvement in precision in results of band-passed CgLp. \\
\begin{table}[t!]
	\centering
	\caption{Parameters of the designed controllers. Frequencies are in Hz.}
	\label{tab:params}
	\resizebox{\columnwidth}{!}{%
		\begin{tabular}{@{}llllllllll@{}}
			\toprule
			Controller       & $\omega_i$ & $\omega_d$ & $\omega_t$ & $\omega_f$ & $\omega_{r}$ & $\gamma$ & $\psi$         & $\lambda$ & $q$  \\ \midrule
			Band-passed CgLp & 10         & 60.6       & 165        & 1000       & 5                  & -0.05        & $-57.34^\circ$ & -0.69     & 2.21 \\
			FORE CgLp        & 10         & 60.6       & 165        & 1000       & 5                  & 0.25     & N/A            & N/A       & N/A  \\
			PID              & 10         & 27.0       & 370        & 1000       & N/A                & N/A      & N/A            & N/A       & N/A  \\ \bottomrule
		\end{tabular}%
	}
\end{table}
\begin{figure}[t!]
	\centering
	\includegraphics[width=\columnwidth]{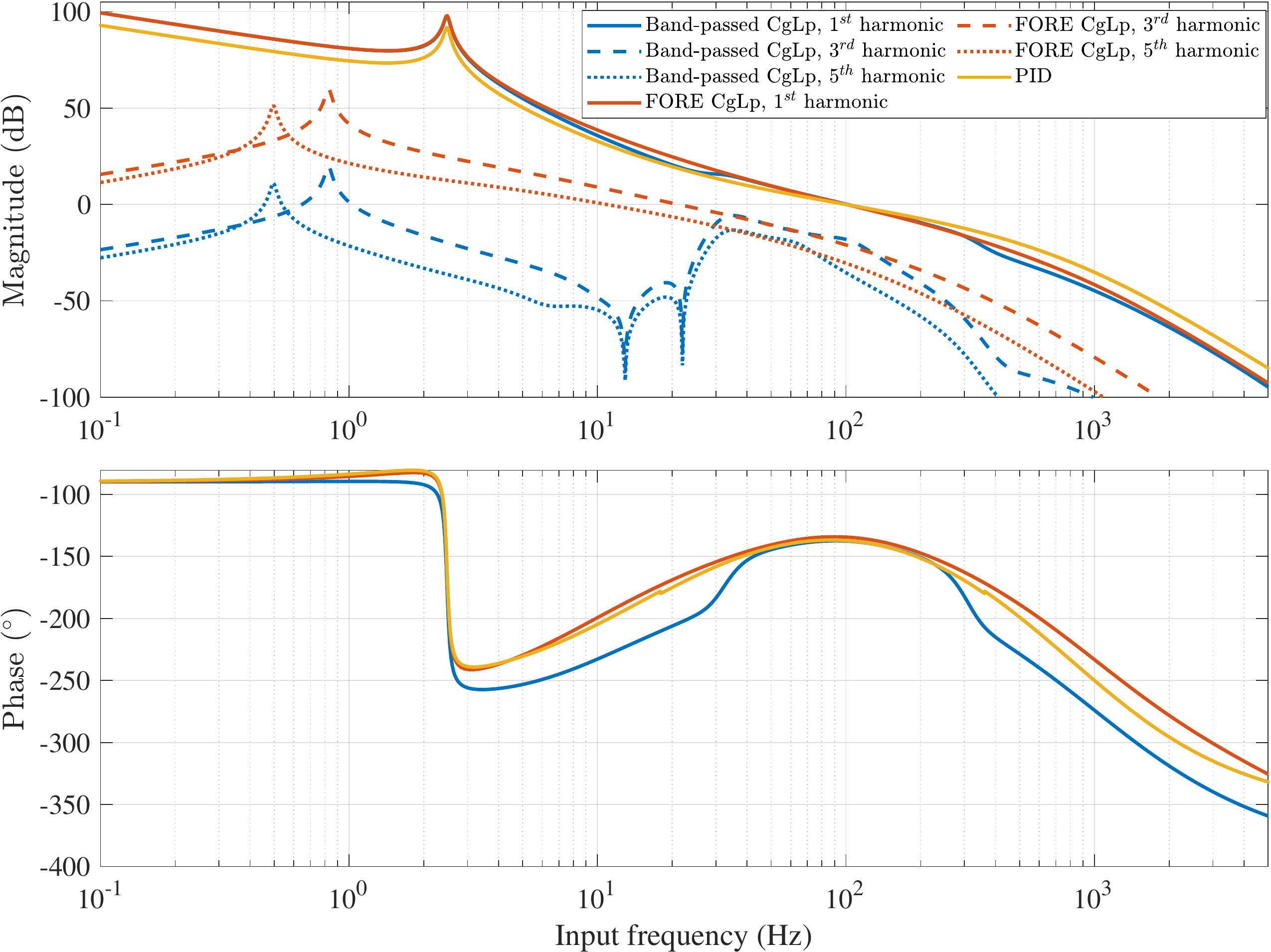}
	\caption{Open loop HOSIDF analysis of the designed controllers including the plant. At lower frequencies, first harmonic gain of band-passed CgLp and FORE CgLp are on top of each other.}
	\label{fig:ol_hosidf}
\end{figure}
\subsection{Practical Implementation}
In order to validate the theories, architectures and methods discussed, the designed controllers have been implemented in practice and their performance has been compared. The implementation was done using National Instruments CompactRIO with a sampling frequency of $10$ kHz.\\
Sinusoidal tracking of different frequencies and different amplitudes have been tested in practice for three designed controllers. Figure~\ref{fig:5hz} shows the error and control input of three controllers to track a 5 Hz sinusoidal input with amplitude of $\num{2e-4}$ m.\\
As it is shown in the Fig.~\ref{fig:5hz}, band-passed CgLp has better steady-state precision than PID as it could be predicted referring to Fig.~\ref{fig:ol_hosidf}. However, conventional CgLp due to presence of higher-order harmonics cannot live up to expectation of first-order DF. The Root Mean Square (RMS) of error for controllers are $\num{3.46e-7}$ m, $\num{1.32e-6}$ m and $\num{4.65e-7}$ m for the band-passed CgLp, the conventional one, and PID, respectively. The figures show a reduction of 25.7\% and 74\% in RMS of steady state error for the band-passed CgLp with respect to PID and the conventional CgLp.\\
Figure~\ref{fig:5hz} also reveals another interesting characteristic of the band-passed CgLp. Reset controllers are known for having large peaks in control input which can saturate the actuator. However, the band-passed CgLp due to limited nonlinearity shows much smaller control input with respect to the conventional CgLp. The maximum of control input for controllers are 0.297 V, 6.980 V, and 0.448 V for the band-passed CgLp, the conventional one, and PID, respectively.\\
\begin{figure}[t!]
	\centering
	\includegraphics[width=\columnwidth]{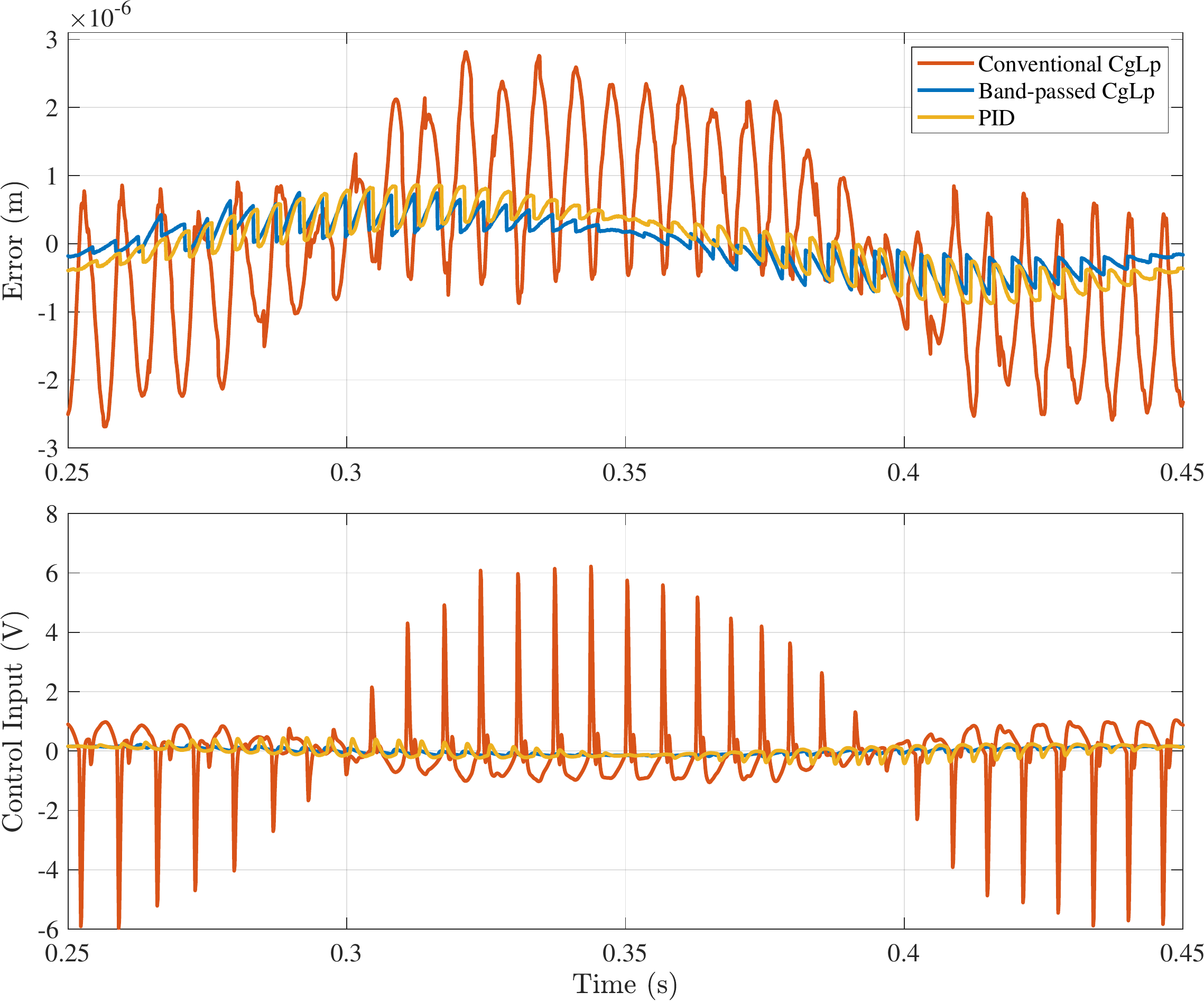}
	\caption{Error and control input of designed controllers, measured in practice for a sinusoidal input of 5 Hz and amplitude of 0.2 mm.}
	\label{fig:5hz}
\end{figure}
The main contribution of this paper and designing band-passed CgLp is to limit nonlinearity to a range of frequencies and reduce it in other frequencies. Thus, it is expected that the band-passed CgLp has lower higher-order harmonics than the conventional one in range of $[0.1~~30]$ Hz. This was also verified in practical implementation as it is shown in single-sided spectrum of Fast Fourier Transform (FFT) of steady-state error for a sinusoidal input of 10~Hz, presented in Fig.~\ref{fig:10HZ}. From the figure it can be observed that the $3^{\text{rd}}$ and the $5^{\text{th}}$ harmonic which are at 30 and 50 Hz, are significantly lower for band-passed CgLp. The same holds for other sinusoidal inputs with frequencies in  $[0.1~~30]$ Hz, however they are not presented for the sake of brevity.\\
\begin{figure}[t!]
	\centering
	\includegraphics[width=\columnwidth]{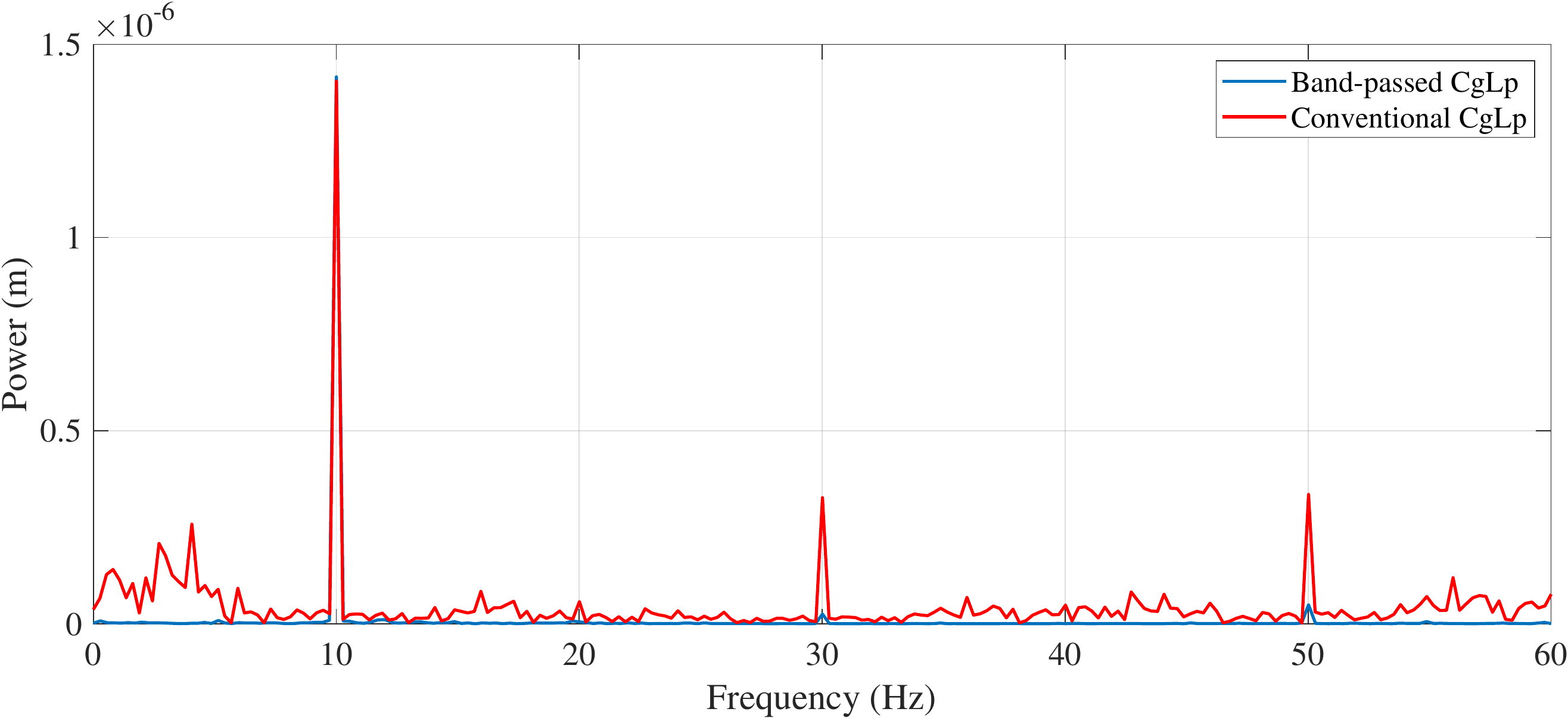}
	\caption{The single-sided FFT spectrum of steady-state error for a sinusoidal input of 10 Hz. The amplitude of reference is \num{7.143e-5} m. }
	\label{fig:10HZ}
\end{figure} 
Furthermore, in order to have a broader view of tracking performance of designed controllers, $L_2$ and $L_\infty$ norm of their steady-state error for sinusoidal inputs of amplitude \num{7.143e-5}~m and frequencies of 1 to 24~Hz in steps of 1 Hz is depicted in Fig.~\ref{fig:tracking_error}. The figure clearly shows a significant decrease of steady-state error for band-passed CgLp with respect to conventional one, indicating the adverse effect of higher-order harmonics in lower frequencies for tracking precision. Due to large peaks present in control input of the conventional CgLp, tracking sinusoidal waves of frequencies larger than 10 Hz was not possible due to actuator saturation. Moreover, $L_\infty$ norm (RMS) of error for band-passed CgLp is lower than PID in almost the entire frequency range till 17~Hz. By resorting to Fig.~\ref{fig:ol_hosidf}, one can notice that from 17~Hz, the higher-order harmonics will increase for band-passed CgLp and decrease again at 22~Hz. The same trend holds for Fig.~\ref{fig:tracking_error}. At very low frequencies, i.e., 1 to 3 Hz, higher-order harmonics of band-passed CgLp is relatively high and thus the steady-state error. A possible suggestion to improve performance at these frequencies is designing shaping filter such that a frequency within this range, e.g., 2 Hz is included in $\omega_{lb}$. This will reduce the higher-order harmonics in this range. \\ 
Remark~\ref{rem:1} suggests that for every frequency in $\omega_{lb}$, higher-order harmonics will be zero. However, in practice due to practical challenges like discretization, quantization and delay, it is expected that this claim does not hold completely. In other words, one can expect a decrease in higher-order harmonics to drop for frequencies in $\omega_{lb}$. Table~\ref{tab:harmonics_22} presents the $1^{\text{st}}$, $3^{\text{rd}}$, $5^{\text{th}}$ and $7^{\text{th}}$ harmonics of steady-state error for sinusoidal inputs of 21, 22, and 23 Hz, where 22 Hz is in  $\omega_{lb}$. The harmonics are obtained using FFT method. The significant drop in higher-order harmonics is observable for 22 Hz. \\
\begin{table}[t!]
	\centering
	\caption{Harmonics of steady-state error for sinusoidal inputs of 21, 22 and 23 Hz. Columns 2 through 5 shows harmonics in m and in columns 5 till 8, harmonics are normalised by the amplitude of input and presented in dB.}
	\label{tab:harmonics_22}
	\resizebox{\columnwidth}{!}{%
		\begin{tabular}{@{}cllllllll@{}}
			\toprule
			Freq. (Hz)& $1^{\text{st}}$ (m)  & $3^{\text{rd}}$ (m)  & $5^{\text{th}}$ (m)  & $7^{\text{th}}$ (m)  & $1^{\text{st}}$ (dB) & $3^{\text{rd}}$ (dB) & $5^{\text{th}}$ (dB) & $7^{\text{th}}$ (dB) \\ \midrule
			21  & \num{5.45e-6} & \num{4.54e-7} & \num{1.97e-7} & \num{7.23e-8} & -22.33        & -43.93        & -51.15        & -59.89        \\
			22  & \num{6.54e-6} & \num{1.92e-7} & \num{1.11e-7} & \num{3.80e-8} & -20.75        & -51.40        & -56.10        & -65.48        \\
			23  & \num{7.50e-6} & \num{1.64e-6} & \num{5.33e-7} & \num{1.22e-7} & -19.56        & -32.77        & -42.52        & -55.31        \\ \bottomrule
		\end{tabular}%
	}
\end{table}
\begin{figure}[t!]
	\centering
	\includegraphics[width=\columnwidth]{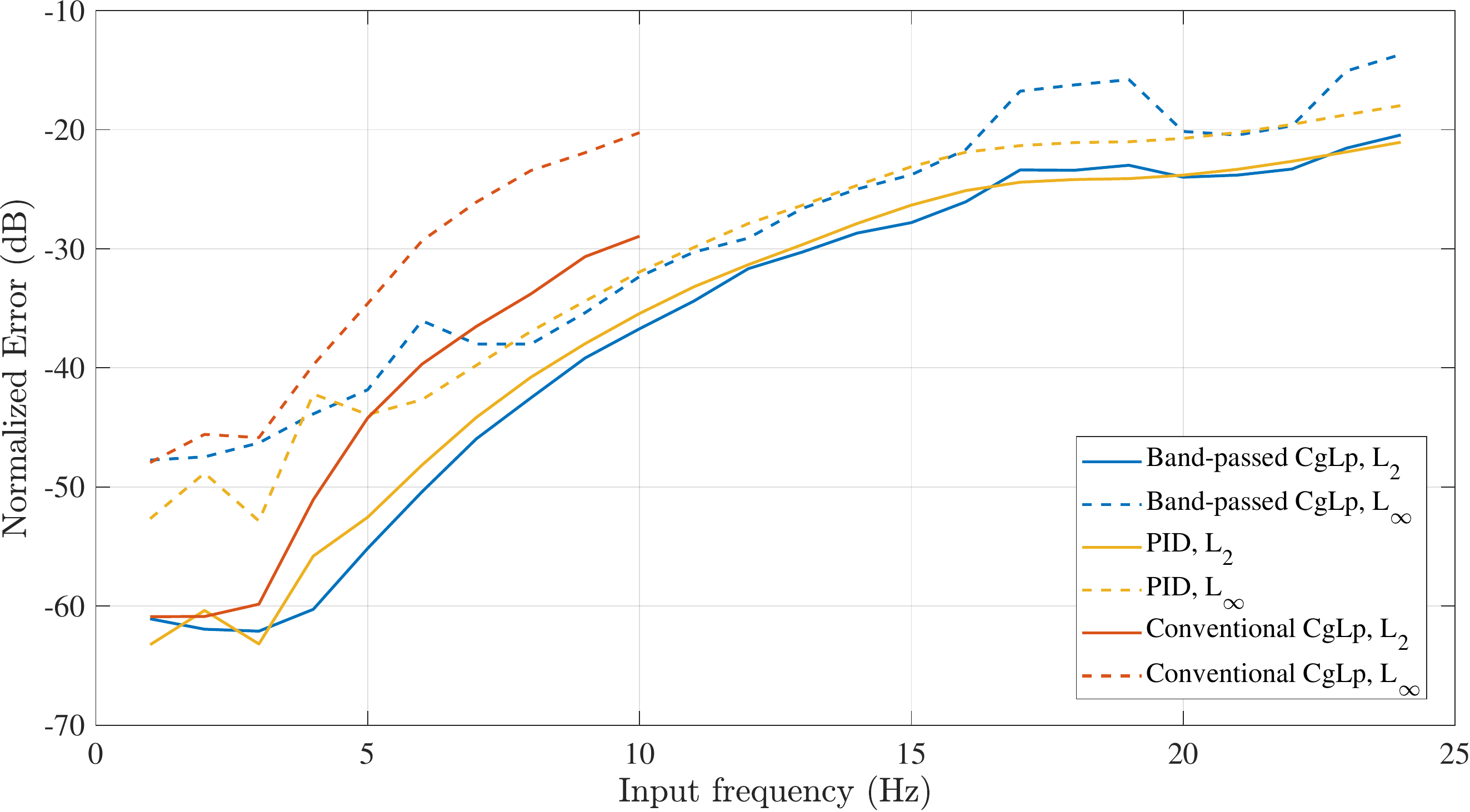}
	\caption{$L_2$ and $L_\infty$ norm of normalized steady-state error for sinusoidal inputs with amplitude of \num{7.143e-5} m and frequencies of 1 to 24 Hz with step of 1 Hz. The error is normalized with respect to amplitude of input.    }
	\label{fig:tracking_error}
\end{figure}
At last, in order to evaluate the performance of the proposed band-passed CgLp controller for multi-sinusoidal tracking, an input constituted of 3 sinusoidal wave was used. The reference which was used is 
\begin{align}\nonumber
	&r(t)=\\&\num{1.5e-5}\sin(2\pi 13 t)+ \num{2.5e-5}\sin(2\pi 7 t)+ \num{5.0e-5}\sin(5\pi 5 t).
	\label{eq:r_3_sin}
\end{align}
Figure~\ref{fig:multi_sin} shows the error and control input for three designed controllers. The band-passed CgLp still shows less steady-state error with respect to other controllers and no large peak in control input.\\
\begin{figure}[t!]
	\centering
	\includegraphics[width=\columnwidth]{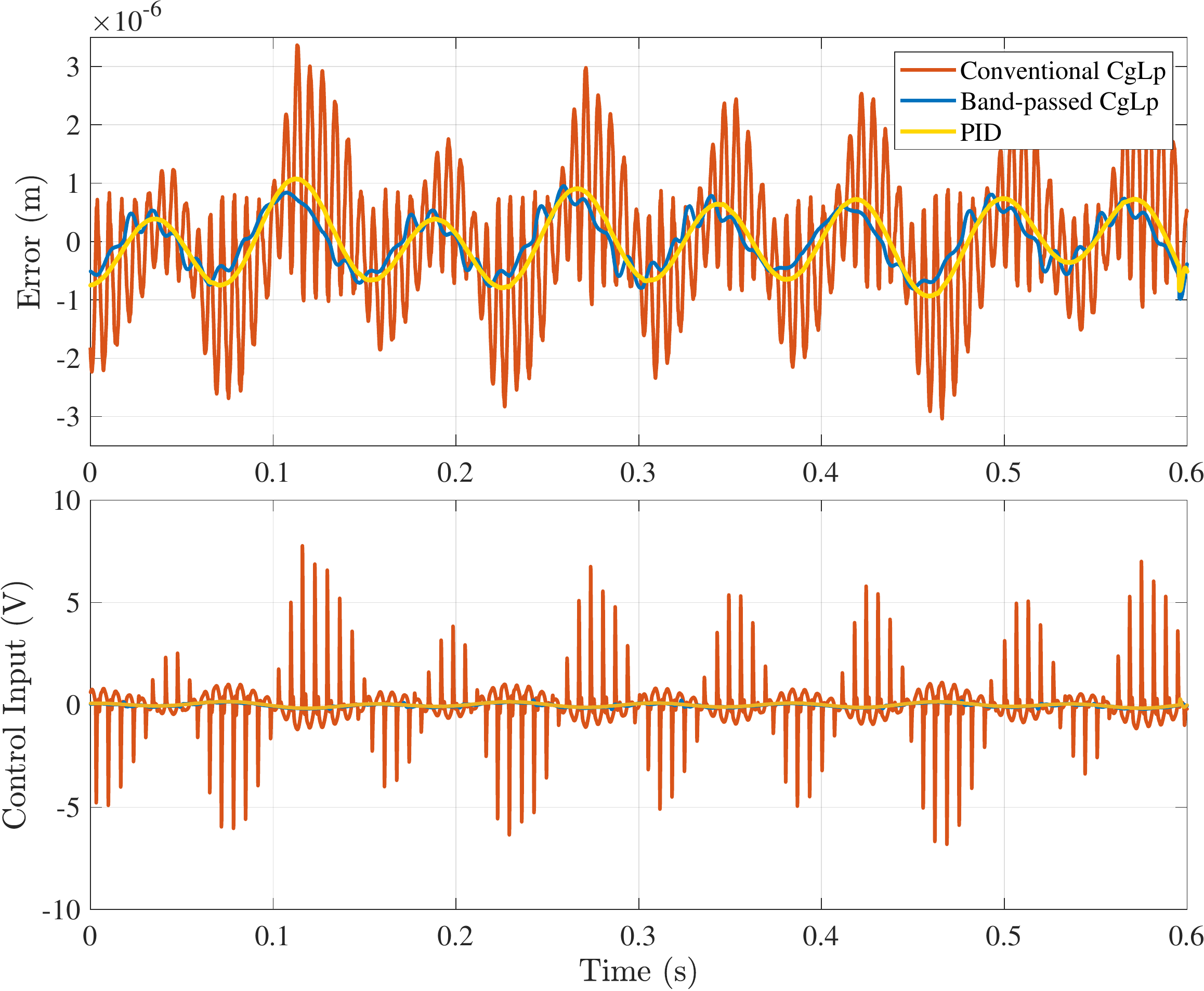}
	\caption{Error and control input compared for 3 designed controllers for the input of Eq.~(\ref{eq:r_3_sin}). The RMS of error is $\num{4.63e-7}$ m, $\num{1.17e-6}$ m and $\num{5.35e-7}$ m for band-passed CgLp, conventional CgLp and PID, respectively. The maximum of steady-state error for controllers is $\num{1.13e-6}$ m, $\num{3.47e-6}$ m and $\num{1.18e-6}$ m.}
	\label{fig:multi_sin}
\end{figure}

\section{Conclusion}
This paper investigated the nonlinearity and higher-order harmonics for reset elements with one resetting state. A new architecture was introduced which allowed for band-passing nonlinearity in a range of frequency and selectively reducing higher-order harmonics in a range of frequencies. After developing the HOSIDF analysis of the proposed architecture a method called ``phase shaping'' was proposed for design and tune of the introduced architecture. It was shown that first-order reset elements such as Clegg integrator, FORE or CgLp can be band-passed using the proposed architecture and method.\\
It was discussed that nonlinearity and higher-order harmonics can  be beneficial in some range of frequencies such a cross-over frequency region for increasing the phase margin and can be harmful at others like lower frequencies. In the phase shaping method, the approach to eliminate nonlinearity at one frequency was also introduced which is useful for systems with single important working frequency.\\
In order to validate the architecture, method and developed theories, 3 controllers designed to control a precision positioning stage. The controllers were a band-passed CgLp, a conventional one and a PID. It was validated in practice that higher-order harmonics for band-passed CgLp at lower frequencies is much smaller than the conventional one. Moreover, it was shown that there is clear relation between reduction of higher-harmonics at lower frequencies and tracking precision of the system. It was verified in practice that by band-passing higher-order harmonics a CgLp can have the same bandwidth and phase margin as PID and improved tracking precision.\\
Since the phase shaping method is capable of shaping the phase benefit of reset element, one may suggest shaping the phase benefit to achieve other characteristics such as iso-damping behaviour for the system or constant gain and positive phase slope. Furthermore, in this paper, only the band-passed CgLp was studied in detail, investigation of band-passed Clegg and FORE are considered as future works.


%



\section*{Acknowledgment}
This work was supported by NWO, through OTP TTW project \#16335.


\ifCLASSOPTIONcaptionsoff
  \newpage
\fi



\bibliographystyle{IEEEtran}
\bibliography{IEEEabrv,ref}
%
%
%

%

\begin{IEEEbiography}[{\includegraphics[width=1in,height=1.25in,clip,keepaspectratio]{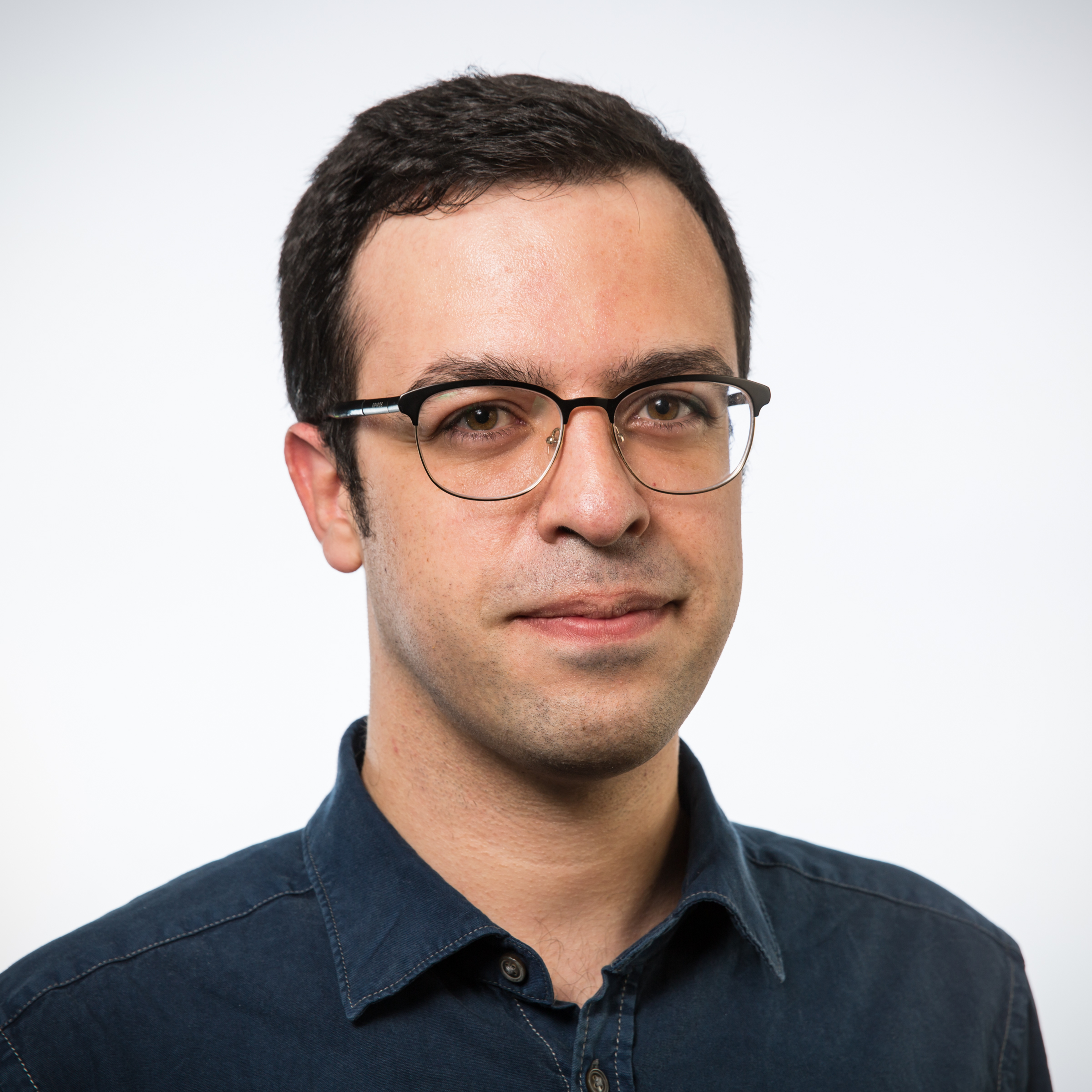}}]{Nima Karbasizadeh}
has received his M.Sc. degree in Mechatronics from University of Tehran, Iran in 2017. He is currently a PhD candidate at department of precision and microsystem engineering, Delft University of Technology, the Netherlands. His research interests are precision motion control, nonlinear precision control, mechatronic system design and haptics. 
\end{IEEEbiography}
\begin{IEEEbiography}[{\includegraphics[width=1in,height=1.25in,clip,keepaspectratio]{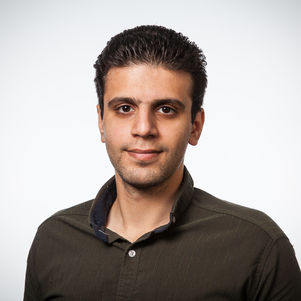}}]{Ali Ahmadi Dastjerdi}
received his master degree in mechanical engineering from Sharif University of Technology, Iran, in 2015. He is currently working as a PhD candidate at the department of precision and microsystem engineering, TU Delft, The Netherlands. His primary research interests are  mechatronic systems design, precision engineering, precision motion control, and nonlinear control.
\end{IEEEbiography}
\begin{IEEEbiography}[{\includegraphics[width=1in,height=1.25in,clip,keepaspectratio]{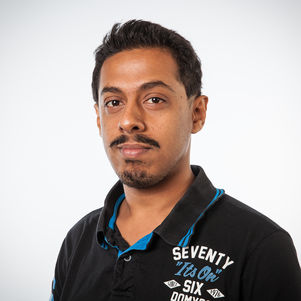}}]{Niranjan Saikumar}
 received his PhD degree in electrical engineering from
Indian Institute of Science, India in 2015. He is currently working as a postdoc
at the department of precision and microsystem engineering, TU Delft, The
Netherlands. His research interests are precision motion control, and
nonlinear precision control and mechatronic system with distributed actuation.
\end{IEEEbiography}
\begin{IEEEbiography}[{\includegraphics[width=1in,height=1.25in,clip,keepaspectratio]{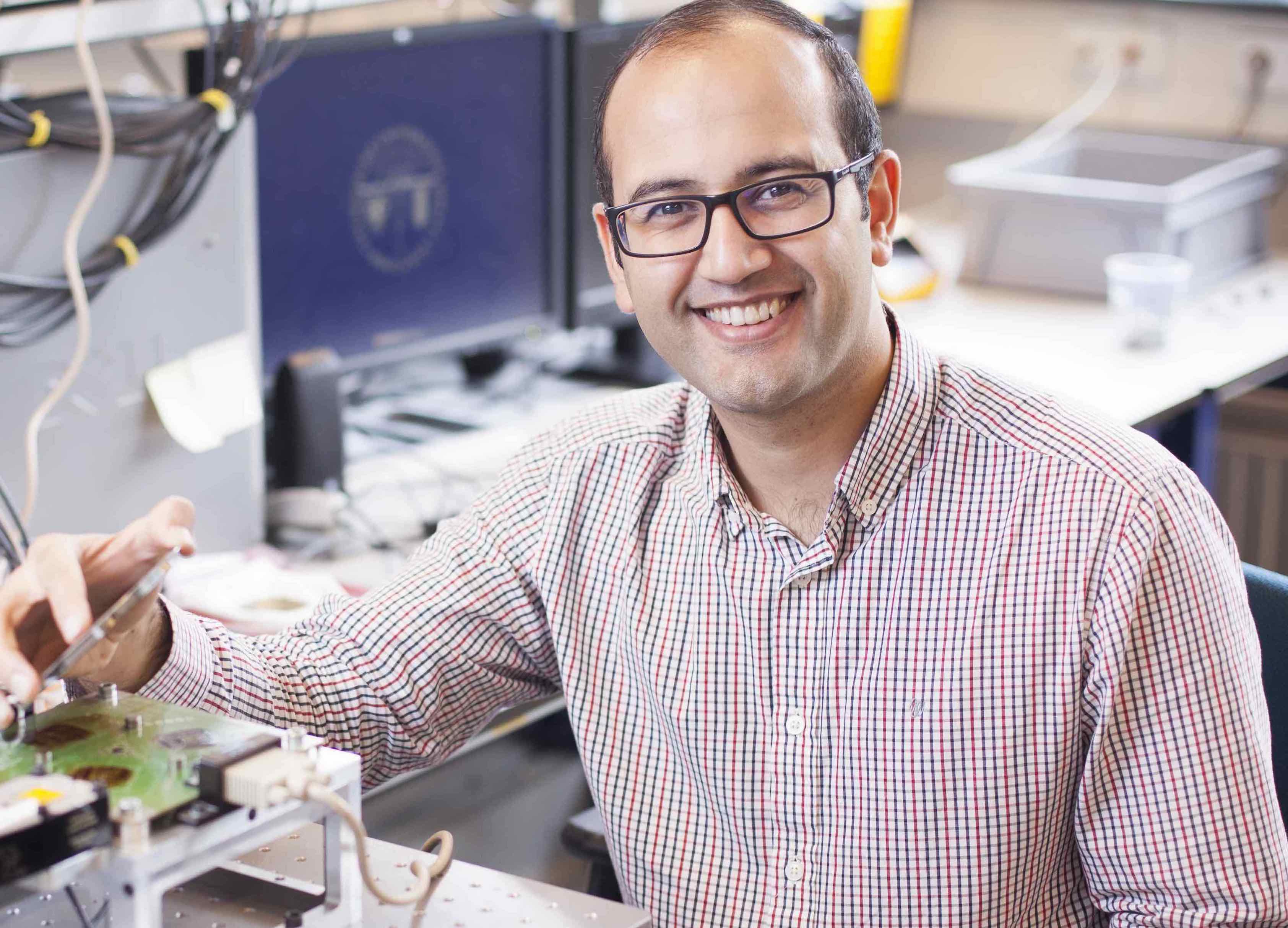}}]{S.~Hassan~HosseinNia}
received his PhD degree with honor "cum laude" in
electrical engineering specializing in automatic control: application in
mechatronics, form the University of Extremadura, Spain in 2013. His main
research interests are precision mechatronic system design, precision
motion control and mechatronic system with distributed actuation and
sensing. He has an industrial background working at ABB, Sweden. Since
October 2014 he is appointed as an assistant professor at the department of
precession and microsystem engineering at TU Delft, The Netherlands. He is
an associate editor of the international journal of advanced robotic systems
and Journal of Mathematical Problems in Engineering.
\end{IEEEbiography}
%
%
%




\end{document}